\documentclass[reqno, eucal]{amsart}

\usepackage[mathscr]{eucal} % \EuScript (\mathcal unchanged)
\usepackage{graphicx}       
\usepackage{epstopdf}      
\usepackage{xcolor}         
\usepackage{amsmath, amsfonts, amssymb, amsthm, amscd, mathabx}
\usepackage{epic, eepic}
\usepackage{longtable, array,subcaption}
\usepackage{here}
\usepackage{pdfpages}

\usepackage[colorlinks=true, pdfstartview=FitV, linkcolor=blue, citecolor=blue, urlcolor=blue]{hyperref}

\newcommand{\arxiv}[1]{\href{https://arxiv.org/abs/#1}{\texttt{arXiv:#1}}}

\usepackage{geometry}
\newtheorem{theorem}{Theorem}
\newtheorem{lemma}[theorem]{Lemma}

\newtheorem{proposition}[theorem]{Proposition}
\newtheorem{corollary}[theorem]{Corollary}

\theoremstyle{definition}
\newtheorem{definition}[theorem]{Definition}
\newtheorem{example}[theorem]{Example}
\newtheorem{remark}[theorem]{Remark}

\newcommand{\Z}{{\mathbb Z}}

\newcommand{\am}{{\rm\aaa}^{\!-} }
\newcommand{\ap}{{\rm\aaa}^{\!+} }
\newcommand{\apm}{{\rm\aaa}^{\!\pm} }

\newcommand{\eb}{\rm{\bf e}}
\newcommand{\F}{\mathcal{F}}
\newcommand{\ot}{\otimes}
\newcommand{\Q}{{\mathbb Q}}

\newcommand{\ket}[1]{\lvert #1 \rangle}
\newcommand{\abs}[1]{\lvert #1 \rvert}
\newcommand{\prob}{P}
\newcommand{\stprob}{\mathbb{P}}  % stationary probability
\newcommand{\aaa}{\mathbf{a}}
\newcommand{\bb}{\mathbf{b}}
\newcommand{\cc}{\mathbf{c}}
\newcommand{\hh}{\mathbf{h}}
\newcommand{\ii}{\mathbf{i}}
\newcommand{\jj}{\mathbf{j}}
\newcommand{\kk}{\mathbf{k}}
\newcommand{\mm}{\mathbf{m}}

\newcommand{\qoa}{\mathcal{A}}  % quantum oscillator algebra
\DeclareMathOperator{\wt}{wt}
\DeclareMathOperator{\tr}{Tr}

\definecolor{darkred}{rgb}{0.7,0,0} % darkred color
\newcommand{\defn}[1]{{\color{blue}\emph{#1}}} % emphasis of a definition

\usepackage[colorinlistoftodos]{todonotes}

\begin{document}

\title[Strange five vertex model and ASEP]
{A strange five vertex model and multispecies ASEP on a ring}

\author[Atsuo Kuniba]{Atsuo Kuniba}
\address{Atsuo Kuniba, Institute of Physics, Graduate School
of Arts and Sciences, University of Tokyo, Komaba, Tokyo, 153-8902, Japan}
\email{atsuo.s.kuniba@gmail.com}

\author[Masato Okado]{Masato Okado}
\address{Masato Okado, Osaka Central Advanced Mathematical Institute \& 
Department of Mathematics, Osaka Metropolitan University, Osaka, 558-8585, Japan}
\email{okado@omu.ac.jp}

\author[Travis Scrimshaw]{Travis Scrimshaw}
\address{Travis Scrimshaw, Department of Mathematics, Hokkaido University, Sapporo  060-0808, Japan}
\email{tcscrims@gmail.com}

\maketitle

\begin{abstract}
We revisit the problem of constructing the stationary states of the multispecies asymmetric simple exclusion process on a one-dimensional periodic lattice. 
Central to our approach is a quantum oscillator weighted five vertex model 
which features a strange weight conservation distinct from the conventional one. 
Our results clarify the interrelations among several known results and refine their derivations. 
For instance, the stationary probability derived from the multiline queue construction by Martin (2020) and Corteel--Mandelshtam--Williams (2022) is identified with the partition function of a three-dimensional system.
 The matrix product operators by Prolhac--Evans--Mallick (2009) acquire a natural diagrammatic 
 interpretation as corner transfer matrices (CTM).
The origin of their recursive tensor structure, as questioned by Aggarwal--Nicoletti--Petrov (2023), 
is revealed through the CTM diagrams.
Finally, the derivation of the Zamolodchikov--Faddeev algebra by Cantini--de Gier--Wheeler (2015) is made intrinsic 
by elucidating its precise connection to a solution to the Yang--Baxter equation originating from quantum group representations.
\end{abstract}

\section{Introduction}

The asymmetric simple exclusion process (ASEP)~\cite{MGP,Sp} is a fundamental model of non-equi\-lib\-ri\-um stochastic dynamics with many applications in physics, biology, probability theory, and other scientific fields. 
In recent years, it has been extensively studied, particularly in one dimension, leading to a variety of generalizations and a wealth of results that intersect with statistical mechanics, algebraic combinatorics, special functions, integrable systems, 
representation theory, etc.
See for example~\cite{AKSS, ANP, BW, CDW, CMW, CRV, Martin20, PEM} and the references therein.

In this paper, we consider the standard continuous time $n$-species ASEP on a periodic lattice of length $L$. 
Each local state is selected from $\{0,1,\ldots, n\}$, where $1,\ldots, n$ represent the presence of one of the $n$ species of particles, and $0$ corresponds to an empty site. 
The model includes a parameter $t$ that determines the asymmetry of the nearest-neighbor hopping rates. 
The first significant problem is to construct a stationary state, which is unique within each sector specified by the particle content.
The problem is trivial for $n=1$, where all the possible states are equally probable.

The multispecies case $n\ge 2$ is non-trivial and has been solved in two intriguing ways: combinatorially and algebraically. 
The combinatorial approach is known 
as the multiline queue (MLQ) construction~\cite{CMW,Martin20}, 
while the algebraic method is based on the 
Zamolodchikov--Faddeev (ZF) algebra~\cite{CDW,PEM}.
The latter directly leads to the matrix product formula for the (unnormalized) 
stationary probability of the configuration $(\sigma_1,\ldots, \sigma_L) \in \{0,\ldots, n\}^L$:
\[
\mathbb{P}(\sigma_1,\ldots, \sigma_L) = \tr(X_{\sigma_1} \cdots X_{\sigma_L}).
\]
Here $X_0,\ldots, X_n$ are operators acting on some auxiliary space over which the trace is taken.
It is well known that the above formula is valid if there are (not necessarily unique)
spectral parameter dependent versions
$X_0(z), \ldots, X_n(z)$ satisfying the ZF algebra whose structure function 
is a stochastic $R$ matrix related to the Markov matrix of the $n$-ASEP.

Central in our approach is a certain five vertex model 
on the two-dimensional square lattice whose ``Boltzmann weights" 
take values in a $t$-deformed quantum oscillator algebra acting on its bosonic Fock space. 
We call it the $t$-oscillator weighted five vertex model. See (\ref{t:s5V}). 
A curious feature is that it does {\em not} satisfy the usual weight conservation
or the so-called  ``ice condition" as in the six vertex model~\cite{Bax}, nor does it satisfy the Yang--Baxter equation.
However, the model can also be interpreted as a three-dimensional system, where the Fock space is attached to the edges in the third direction.

The strange five vertex model plays a pivotal role, refining many known results~\cite{CDW,CMW,Martin20,PEM} and synthesizing their techniques together.
Despite not satisfying the Yang--Baxter equation, it still presents many beneficial aspects:
\begin{itemize}
\item The operators $X_0(z),\ldots, X_n(z)$ are formulated as \emph{corner transfer matrices} (CTM)  \`a la Baxter (cf.~\cite[Chap.13]{Bax}).
\item These CTM diagrams immediately lead to the recursion relation of these operators with respect to~$n$.
\item Stationary probabilities are identified with the partition functions of a three-dimensional system.
\item The generating sums of the combinatorial weights in MLQs are readily identified with a simple trace of the $t$-oscillators.
\item The rank-reducing $RTT=TTR$ relation for proving the ZF algebra is linked with the standard solution of the Yang--Baxter equation constructed from the symmetric tensor representations of the quantum group (see Remark~\ref{a:re:L}).
\end{itemize}

The second item above clarifies the origin of the tensor structure of 
$X_\alpha$ questioned at the end of~\cite[Sec.4.2]{ANP}. 
The conservation law in the strange five vertex model is designed 
around the MLQ pairing, providing a natural framework for such analysis. 
This paper presents these results concisely, without the need for heavy 
machinery from integrable probability, making it both accessible and efficient to read.

The outline of the paper is as follows.
In Section~\ref{sec:asep}, we recall the $n$-ASEP and a general relation between 
a matrix product formula $\tr(X_{\sigma_1}\cdots X_{\sigma_L})$ 
for the stationary probabilities and the ZF algebra among the 
spectral parameter-dependent operators 
$X_0(z),\ldots, X_n(z)$.

In Section~\ref{sec:mlq}, we reexamine  the combinatorial approach to stationary states 
by the MLQ construction~\cite{CMW,Martin20}. 
We demonstrate that the result can be expressed 
as a certain composition $\mathbb{M}(q,t)$ of linear operators 
which naturally lends itself to a diagrammatic representation as a CTM of size $n$ 
(Proposition~\ref{a:pr:P} and (\ref{eq:f8})).

In Section~\ref{sec:5v}, we introduce the strange five vertex model, whose 
statistical weights take values in a $t$-oscillator algebra $\qoa$.
This model can be interpreted as a three dimensional (3D) system, 
where $\qoa$ acts on Fock spaces in the third dimension.
We establish that the stationary probabilities correspond to the partition function of this 3D system, 
with boundary conditions derived from the $n$-ASEP configuration 
(Theorem~\ref{t:th:ms} and (\ref{a:fig10})).
The matrix product operators $X_0(z),\ldots, X_n(z)$ play the  role of  
the layer transfer matrices of the system (\ref{a:ef12}).
They obey a recursion relation with respect to the rank $n$,
 which follows directly from the CTM diagrams.

In Section~\ref{sec:zf}, we provide a new proof of 
the ZF algebra relation among $X_0(z),\ldots, X_n(z)$.
It is the most natural one from the perspective of quantum integrable systems, 
elucidating a precise relationship (\ref{a:cL}) with the Yang--Baxter equation of the relevant quantum $R$ matrices.

Section~\ref{sec:cr} is devoted to concluding remarks.

After the basic definitions of the model in Section~\ref{sec:asep}, 
the text can also be read in the following order: 
Section~\ref{sec:5v} and  Section~\ref{sec:zf}, 
to establish the matrix product formula first before proceeding 
to Section~\ref{sec:mlq}, where the connection with the MLQ method is explained.

\section{Multispecies ASEP}\label{sec:asep}

\subsection{Definition of $n$-ASEP}
Consider the periodic 1D lattice with $L$ sites, which will be denoted by $\Z_L$.
Each site $i \in \Z_L$ is assigned with a variable
$\sigma_i \in \{0,1,\ldots, n\}$, where 
$\sigma_i = \alpha$ is interpreted that the site $i$ is occupied by a particle of type $\alpha$
if $\alpha \neq 0$ and vacant if $\alpha=0$.
We assume $1 \le n <  L$ throughout.
The space of states is given by
\begin{align}\label{W}
(\mathbb{C}^{n})^{\otimes L} \simeq 
\bigoplus_{(\sigma_1,\ldots, \sigma_L) \in\{0,\ldots, n\}^L} 
\mathbb{C}|\sigma_1,\ldots, \sigma_L\rangle.
\end{align}
Consider a stochastic process
in which neighboring pairs of local states
$(\sigma_i, \sigma_{i+1}) = (\sigma, \sigma')$ 
are interchanged as $(\sigma,\sigma') \rightarrow (\sigma', \sigma)$
with the transition rate $t^{\theta(\sigma<\sigma')}$ with some parameter $t\ge 0$.
Here and in what follows we use the notation $\theta(\text{true}) = 1$ and $\theta(\text{false}) = 0$.
Let $P(\sigma_1,\ldots, \sigma_L; T)$ be the probability of finding 
the state $\ket{\sigma_1,\ldots, \sigma_L}$ at time $T$, and set
\begin{align}
|P(T)\rangle
= \sum_{(\sigma_1,\ldots, \sigma_L) \in\{0,\ldots, n\}^L}
P(\sigma_1,\ldots, \sigma_L; T)|\sigma_1,\ldots, \sigma_L\rangle.
\end{align}
By $n$-ASEP we mean a Markov process
governed by the continuous-time master equation
\begin{align}
\frac{d}{dT}|P(T)\rangle
= H |P(T)\rangle,
\end{align}
where the Markov matrix\footnote{Also called a ``Hamiltonian'' by abuse of terminology despite it not being Hermitian in general.}
has the form
\begin{align}
H &= \sum_{i \in \Z_L}H^{loc}_{i,i+1},
\qquad 
H^{loc} \colon \ket{\alpha, \beta} \mapsto 
(\ket{\beta,\alpha} - \ket{\alpha, \beta}) t^{\theta(\alpha< \beta)},
\end{align}
where $H^{loc}_{i,i+1}$ acts on the $i$th and the $(i\!+\!1)$th components
as $H^{loc}$ and as the identity elsewhere.
By the definition $H^{loc}$ is expressed as
\begin{align}\label{hloc}
H^{loc} &= \sum_{0\le \alpha < \beta \le n}
(t E_{\beta \alpha}\otimes E_{\alpha \beta} -t E_{\alpha \alpha}\otimes E_{\beta \beta}
+ E_{\alpha \beta} \otimes E_{\beta \alpha} - E_{\beta \beta}\otimes E_{\alpha \alpha})
\end{align}
in terms of the matrix unit $E_{\alpha \beta}$ acting as
$E_{\alpha \beta} \ket{\gamma} = \delta_{\beta \gamma} \ket{\alpha}$.

\begin{example}
Consider $n = 1$ and $L = 3$, then we have
\[
H^{loc}_{1,2} = \begin{pmatrix}
0
\\ & 0 & 0 & 0
\\ & 0 & -t & 1
\\ & 0 & t & -1
\\ &&&& -t & 1 & 0
\\ &&&& t & -1 & 0
\\ &&&& 0 & 0 & 0
\\ &&&&&&&0
\end{pmatrix}
\begin{matrix}
\ket{000} \\ \ket{001} \\ \ket{010} \\ \ket{100} \\ \ket{011} \\ \ket{101} \\ \ket{110} \\ \ket{111}
\end{matrix},
\quad
H = \begin{pmatrix}
0
\\ & A & 1 & t
\\ & t & A & 1
\\ & 1 & t & A
\\ &&&& A & 1 & t
\\ &&&& t & A & 1
\\ &&&& 1 & t & A
\\ &&&&&&&0
\end{pmatrix},
\]
where $A = -t-1$ and $H = H_{1,2}^{loc} + H_{2,3}^{loc} + H_{3,1}^{loc}$ with noting $H_{3,1}^{loc} = H_{3,4}^{loc}$ by convention.
\end{example}

As $H$ preserves the particle content, 
it acts on each \defn{sector} labeled with 
the \defn{multiplicity} $\mm=(m_0, \ldots, m_n) \in (\Z_{\ge 0})^{n+1}$ of the particles:
\begin{align}\label{a:VP}
W(\mm) =\!\!
\sum_{{\boldsymbol \sigma} \in 
\Sigma(\mm)}\! \mathbb{C}\ket{{\boldsymbol \sigma}}, \quad
\Sigma(\mm) = 
\{{\boldsymbol \sigma}=(\sigma_1, \ldots, \sigma_L) \in \{0, \ldots, n\}^L\; \mid \;
\sum_{j=1}^L \delta_{\alpha,\sigma_j} = m_\alpha,\forall \alpha\}. 
\end{align}
Note that $m_0 + \cdots + m_n = L$ holds and $\dim W(\mm) = \frac{L!}{m_0! \cdots m_n!}$.
A sector $W(m_0,\ldots, m_n)$ such that $m_\alpha \ge 1$ for all $0 \le \alpha \le n$ is called \defn{basic}.
Non-basic sectors are equivalent to a basic sector for $n'$-ASEP with some $n' < n$ by a suitable relabeling of species.
Thus we shall exclusively deal with basic sectors in this paper
(hence $n \le L$ as mentioned before).

\begin{example}
\label{ex:markov_matrix}
Consider $n = 2$ and $L = 4$.
Then the matrix $H$ restricted to the sector $\mm = (2,1,1)$ is
\[
\setcounter{MaxMatrixCols}{20}
\begin{pmatrix}
A & 1 & {\color{black!20}0} & {\color{black!20}0} & {\color{black!20}0} & 1 & {\color{black!20}0} & {\color{black!20}0} & {\color{black!20}0} & {\color{black!20}0} & t & {\color{black!20}0}  \\
t & B & 1 & 1 & {\color{black!20}0} & {\color{black!20}0} & {\color{black!20}0} & {\color{black!20}0} & {\color{black!20}0} & {\color{black!20}0} & {\color{black!20}0} & t \\
{\color{black!20}0} & t & C & {\color{black!20}0} & 1 & {\color{black!20}0} & {\color{black!20}0} & {\color{black!20}0} & {\color{black!20}0} & t & {\color{black!20}0} & {\color{black!20}0} \\
{\color{black!20}0} & t & {\color{black!20}0} & A & 1 & {\color{black!20}0} & {\color{black!20}0} & {\color{black!20}0} & 1 & {\color{black!20}0} & {\color{black!20}0} & {\color{black!20}0} \\
{\color{black!20}0} & {\color{black!20}0} & t & t & B & 1 & 1 & {\color{black!20}0} & {\color{black!20}0} & {\color{black!20}0} & {\color{black!20}0} & {\color{black!20}0} \\
t & {\color{black!20}0} & {\color{black!20}0} & {\color{black!20}0} & t & C & {\color{black!20}0} & 1 & {\color{black!20}0} & {\color{black!20}0} & {\color{black!20}0} & {\color{black!20}0} \\
{\color{black!20}0} & {\color{black!20}0} & {\color{black!20}0} & {\color{black!20}0} & t & {\color{black!20}0} & A & 1 & {\color{black!20}0} & {\color{black!20}0} & {\color{black!20}0} & 1 \\
{\color{black!20}0} & {\color{black!20}0} & {\color{black!20}0} & {\color{black!20}0} & {\color{black!20}0} & t & t & B & 1 & 1 & {\color{black!20}0} & {\color{black!20}0} \\
{\color{black!20}0} & {\color{black!20}0} & {\color{black!20}0} & t & {\color{black!20}0} & {\color{black!20}0} & {\color{black!20}0} & t & C & {\color{black!20}0} & 1 & {\color{black!20}0} \\
{\color{black!20}0} & {\color{black!20}0} & 1 & {\color{black!20}0} & {\color{black!20}0} & {\color{black!20}0} & {\color{black!20}0} & t & {\color{black!20}0} & A & 1 & {\color{black!20}0} \\
1 & {\color{black!20}0} & {\color{black!20}0} & {\color{black!20}0} & {\color{black!20}0} & {\color{black!20}0} & {\color{black!20}0} & {\color{black!20}0} & t & t & B & 1 \\
{\color{black!20}0} & 1 & {\color{black!20}0} & {\color{black!20}0} & {\color{black!20}0} & {\color{black!20}0} & t& {\color{black!20}0} & {\color{black!20}0} & {\color{black!20}0} & t & C
\end{pmatrix}
\begin{matrix}
\ket{0012} \\ \ket{0102} \\ \ket{1002} \\
\ket{0120} \\ \ket{1020} \\ \ket{0021} \\
\ket{1200} \\ \ket{0201} \\ \ket{0210} \\
\ket{2001} \\ \ket{2010} \\ \ket{2100}
\end{matrix}
\]
where $A = -2t-1$, $B = -2t-2$, and $C = -t-2$.
Note that the left null eigenvector of the Markov matrix $H$ is $(1, \ldots, 1)$ reflecting the total probability conservation.
\end{example}

\subsection{Stationary states}

In each sector $W(\mm)$ there is a unique state 
$\ket{\overline{P}(\mm)}$ up to a normalization, called the \defn{stationary state},  
satisfying $H\ket{\overline{P}(\mm)} = 0$.

The stationary state for $1$-TASEP is uniform in that all the configurations are realized with an equal probability.

\begin{example}\label{ex:pbar}
We present (unnormalized) steady states in small sectors of 
$2$-ASEP and $3$-ASEP in the form
\begin{align*}
 \ket{\overline{P}(\mm)} = \ket{\xi(\mm)} + C\ket{\xi(\mm)}
 + \cdots + C^{L-1}\ket{\xi(\mm)}
 \end{align*}
 where $C$ denotes a cyclic shift 
 $C\ket{\sigma_1,\ldots, \sigma_L} = \ket{\sigma_L, \sigma_1, \ldots, \sigma_{L-1}}$.
 Note that the choice of $\ket{\xi(\mm)}$ is not unique.
\begin{align*}
\ket{\xi(1,1,1)} &= (2+t)\ket{012} + (1+2t)\ket{021},
\\
\ket{\xi(2,1,1)} &=(3+t)\ket{0012} + 2(1+t)\ket{0102}+(1+3t)\ket{1002},
\\
\ket{\xi(1,2,1)} &=(2+t+t^2)\ket{0112}+(1+t)^2\ket{1012} + (1+t+2t^2)\ket{1102},
\\
\ket{\xi(1,1,2)} &= (3+t)\ket{1220} +2(1+t) \ket{2120} + (1+3t)\ket{2210},
\\
\ket{\xi(1,2,2)} & =(3+t+t^2) \ket{11220} + (2+t+2t^2) \ket{12120} + 
(1+3t+t^2)\ket{12210} 
\\
&+ (2+t+2t^2) \ket{21120} +(1+2t+2t^2) \ket{21210} + (1+t+3t^2)\ket{22110},\\
 \ket{\xi(2,1,2)} & =(1 + 6 t + 7 t^2 + 6 t^3)\ket{ 00221} 
 + (2 + 7 t + 6 t^2 + 5 t^3) \ket{ 02021} + (1 + t) (3 + 4 t + 3 t^2)\ket{ 02201} 
 \\
 &+ (1 + t) (3 + 4 t + 3 t^2) \ket{ 20021} + (5 + 6 t + 7 t^2 + 2 t^3) \ket{ 20201} 
 + (6 + 7 t + 6 t^2 + t^3) \ket{ 22001},\\
 \ket{\xi(2,2,1)} & =(3 + t + t^2) \ket{ 00112} + (2 + 2 t + t^2) \ket{ 01012} 
 + (2 + t + 2 t^2) \ket{ 01102}
 \\
 & + (1 + 3 t + t^2)\ket{ 10012} +(1 + 2 t + 2 t^2) \ket{10102} + 
 (1 + t + 3 t^2)\ket{ 11002},\\
 \ket{\xi(1,1,1,1)} & =  (9 + 7 t + 7 t^2 + t^3) \ket{0123}  
 + (3 + 11 t + 5 t^2 + 5 t^3) \ket{0213}  
 \\
& + 3 (1 + t)^3\ket{1023}  + 
 (5 + 5 t + 11 t^2 + 3 t^3) \ket{1203}  
 \\
 &+ 3 (1 + t)^3\ket{2013}  
 + (1 + 7 t + 7 t^2 + 9 t^3)\ket{2103}.
\end{align*}
These formulas reduce to~\cite[Ex.2.1]{KMO2} at $t=0$.
The result $ \ket{\xi(1,1,1,1)}$ agrees with the anti-clockwise reading of ~\cite[Fig.1.3]{Martin20} with $q\rightarrow t$, if the local states $1,2,3,\cdot$ therein are replaced by $3,2,1,0$ here, respectively.
Moreover, according to $\ket{\xi(2,1,1)}$ in the above, the right null eigenvector of the Markov matrix $H$ from Example~\ref{ex:markov_matrix}, the nontrivial stationary state up to normalization, is equal to the (column) vector
\[
[3 + t, 2 (1 + t), 1 + 3 t, 3 + t, 2 (1 + t), 1 + 3 t, 3 + t, 2 (1 + t), 1 + 3 t, 3 + t, 2 (1 + t), 1 + 3 t]^T.
\]
 
\end{example}

\subsection{Matrix product construction}\label{a:subsec:mp}

Consider the stationary state
\begin{align}\label{a:srb}
\ket{\overline{P}(\mm)} = \sum_{{\boldsymbol \sigma} \in \Sigma(\mm)}
\stprob({\boldsymbol \sigma}) \ket{ {\boldsymbol \sigma} }
\end{align}
and suppose that the stationary probability $\stprob({\boldsymbol \sigma})$ 
is expressed in the matrix product form
\begin{align}\label{a:mho}
\stprob(\sigma_1,\ldots, \sigma_L) 
= \tr (X_{\sigma_1}\cdots X_{\sigma_L})
\end{align}
in terms of some (not necessarily unique) operators $X_0, \ldots, X_n$.
Introduce the notations for the matrix elements of the local Markov matrix (\ref{hloc}) and 
the associated product of $X_i$'s as
\begin{align}\label{a:hdef}
H^{loc}\ket{\alpha, \beta } = \sum_{\gamma,\delta}h^{\gamma,\delta}_{\alpha,\beta}
\ket{\gamma, \delta},\qquad
(hXX)_{\alpha, \beta} := \sum_{\gamma,\delta}h^{\alpha,\beta}_{\gamma,\delta}
X_\gamma X_\delta.
\end{align}
Then we have
\begin{align*}
H  \ket{\overline{P}(\mm)} 
&= \sum_{i \in \Z_L}
\sum_{\boldsymbol \sigma \in \Sigma(\mm)} \stprob(\ldots, \sigma_i, \sigma_{i+1},\ldots) H^{loc}_{i,i+1}
\ket{\ldots, \sigma_i, \sigma_{i+1},\ldots}\\
&= \sum_{i \in \Z_L} \sum_{\boldsymbol \sigma \in \Sigma(\mm)} \sum_{\sigma'_i, \sigma'_{i+1}}
\tr(\cdots X_{\sigma_i}X_{\sigma_{i+1}}\cdots )
h^{\sigma'_i, \sigma'_{i+1}}_{\sigma_i, \sigma_{i+1}}
\ket{\ldots, \sigma'_i, \sigma'_{i+1},\ldots}\\
&= \sum_{\boldsymbol \sigma \in \Sigma(\mm)}\sum_{i \in \Z_L}
\tr(\cdots (hXX)_{\sigma_i, \sigma_{i+1}}\cdots )
\ket{\ldots, \sigma_i, \sigma_{i+1},\ldots}.
\end{align*}
Therefore if there are another set of operators 
${\widehat X}_1, \ldots, {\widehat X}_n$ obeying the \defn{hat relation}
\begin{align}\label{a:hrel}
(hXX)_{\alpha, \beta} = X_\alpha {\widehat X}_\beta - {\widehat X}_\alpha X_\beta,
\end{align}
the vector (\ref{a:srb})  satisfies $H\ket{\overline{P}(\mm)} = 0$ 
thanks to the cyclicity of the trace.
Then (\ref{a:mho}), assuming it is non-zero and finite, must coincide with the actual stationary probability 
up to an overall normalization due to the uniqueness of the stationary state.
Note on the other hand that ${\widehat X}_i$ satisfying the hat relation (\ref{a:hrel})
for a given $X_i$ is not unique.
For example
${\widehat X}_i \rightarrow {\widehat X}_i+ c X_i$ keeps (\ref{a:hrel}) valid.

From~\eqref{hloc} we find the explicit form of~\eqref{a:hrel} as
\begin{align}\label{a:hr1}
t^{\theta(\alpha>\beta)}X_\beta X_\alpha 
-t^{\theta(\alpha<\beta)}X_\alpha X_\beta = 
X_\alpha \widehat{X}_\beta - \widehat{X}_\alpha X_\beta \quad (0 \le \alpha, \beta \le n).
\end{align}
It is easily seen that~\eqref{a:hr1} is satisfied by setting 
\begin{align}
X_\alpha = X_\alpha(1),\qquad 
\widehat{X}_\alpha = (1-t)\left. \frac{dX_\alpha(z)}{dz}\right\rvert_{z=1}
\end{align}
for the operator 
$X_0(z), \ldots, X_n(z)$ involving a spectral parameter $z$ 
provided that they obey the relations
\begin{align}
(x-ty)X_\alpha(y)X_\beta(x) &= (1-t) xX_\alpha(x)X_\beta(y) + (x-y) X_\beta(x)X_\alpha(y)
\quad (0 \le \alpha < \beta \le n),
\label{a:hr21}
\\
[X_\alpha(x), X_\beta(y)] &= [X_\alpha(y), X_\beta(x)]\quad (0 \le \alpha, \beta \le n).
\label{a:hr22}
\end{align}
The relation~\eqref{a:hr21} allows one to interchange the order of 
the spectral parameters $y, x$ into $x,y$ for $\alpha<\beta$.
An analogous relation for $\alpha>\beta$ can be derived 
by combining~\eqref{a:hr22} and~\eqref{a:hr21} as 
\begin{align*}
X_\alpha(y)X_\beta(x) &= X_\beta(x)X_\alpha(y)-X_\beta(y)X_\alpha(x) 
+ X_\alpha(x)X_\beta(y)
\\
&=\left(1-\frac{(1-t)x}{x-ty}\right)X_\beta(x)X_\alpha(y) + 
\left(1-\frac{x-y}{x-ty}\right)X_\alpha(x)X_\beta(y).
\end{align*} 
In this way, one finds that~\eqref{a:hr21} and~\eqref{a:hr22} are presented in the form of a 
\defn{Zamolodchikov--Faddeev (ZF) algebra}:
\begin{align}\label{a:zf}
X_\alpha(y)X_\beta(x) = \sum_{\gamma, \delta=0}^n
R\bigl(y/x\bigr)^{\beta,\alpha}_{\gamma,\delta}X_\gamma(x)X_\delta(y).
\end{align}
Here the structure function is given by 
\begin{align}\label{a:rmat}
R(z)^{\alpha, \alpha}_{\alpha, \alpha} = 1,
\qquad
R(z)^{\alpha, \beta}_{\alpha, \beta} = \frac{(1-z)t^{\theta(\alpha<\beta)}}{1-tz},
\qquad
R(z)^{\beta,\alpha}_{\alpha, \beta} = \frac{(1-t)z^{\theta(\alpha>\beta)}}{1-tz},
\end{align}
for $\alpha \neq \beta$.
The other elements are zero.
This is known as a quantum $R$ matrix for the vector representation of $U_t(\widehat{sl}_{n+1})$~\cite{D86,J1}.
Set  $\widecheck{R}(z) = P R(z)$ with $P$ being the transposition $P(u \otimes v) = v \otimes u$.
This satisfies the Yang--Baxter relation (cf.~\cite{Bax})
\begin{align}\label{a:ybR}
\widecheck{R}_{23}(y)\widecheck{R}_{12}(xy)\widecheck{R}_{23}(x) 
= \widecheck{R}_{12}(x)\widecheck{R}_{23}(xy)\widecheck{R}_{12}(y),
\end{align}
which is the associativity of (\ref{a:zf}).
The $R$ matrix is stochastic in the sense that 
$\sum_{\alpha, \beta}R^{\alpha, \beta}_{\gamma,\delta}(z)=1$ for any $0 \le \gamma, \delta \le n$.

\section{Multiline queue construction}\label{sec:mlq}

We recapitulate the multiline queue (MLQ) construction of ASEP states by~\cite{Martin20} in a form adapted to our conventions. 
See also~\cite{CMW}.
We also reformulate it as a composition of the matrix $\widecheck{M}(z,t)$, which we introduce as the building block of the construction.

\subsection{Ball system}
We use the following notations:
\begin{gather}
\ii = (i_1, \ldots, i_L) \in \{0,1\}^L, \qquad \abs{\ii} = i_1 + \cdots + i_L, \qquad 
\ii \le \jj \overset{\text{def}}{\Longleftrightarrow}
\jj - \ii \in (\Z_{\ge 0})^L, 
\label{a:ii}\\
B_l = \{ \ii = (i_1, \ldots, i_L) \in \{0,1\}^L \mid \abs{\ii} = l \}.
\label{a:B}
\end{gather}

Consider a basic sector $W(\mm)$ for $\mm = (m_0, \ldots, m_n)$, and set
\begin{align}
l_i &= m_i + m_{i+1} + \cdots + m_n \quad (0 \le i \le n), \label{a:lk} \\
B(\mm) &= B_{l_n} \otimes B_{l_{n-1}} \otimes \cdots \otimes B_{l_1}. \label{a:bk}
\end{align} 
We prefer to use $\otimes$ to denote the product of sets rather than $\times$ as we are treating these sets as basis elements (in the sense of Kashiwara's crystal bases~\cite{Kashiwara90}).
Note that $L = l_0 > l_1 > \cdots > l_n \ge 1$ since the sector $W(\mm)$ is assumed to be basic.
Elements of $B(\mm)$ will be referred to as a \defn{ball system}.

Consider a ball system given as $\bb = \bb_n \otimes \cdots \otimes \bb_1 \in B(\mm)$, where 
$\bb_i = (b_{i1}, \ldots, b_{iL}) \in B_{l_i}$.
We identify $\bb$ with a \defn{ball diagram}, which is an $n \times L$ rectangular tableau
in which the box at the $i$th row and the $j$th column contains a ball if $b_{ij} = 1$ and is empty if $b_{ij} = 0$.
Here and in what follows, the rows (resp.\ columns) are numbered from the top 
(resp.\ the left) of the diagram.
The $r$-th row corresponding to $\bb_r$ will simply be called Row $r$.
A ball is understood as carrying the information of its location in the tableau.

\begin{example}\label{a:ex1}
Consider the $3$-ASEP on the length $L = 9$ lattice in the sector 
$W(\mm)$ with  ${\bf m}= (2,3,2,2)$. We have $(l_1, l_2, l_3) = (7,4,2)$.
Consider a ball system
\begin{subequations}
\begin{gather}
\bb = \bb_3 \otimes \bb_2 \otimes \bb_1 \in B_2 \otimes B_4 \otimes B_7, \\
\bb_3 = (001010000), \quad
\bb_2 = (110100010), \quad
\bb_1 = (011111101).
\end{gather}
\end{subequations}
Its corresponding ball diagram looks as follows:
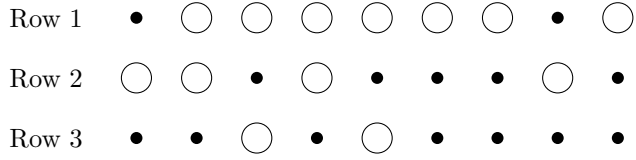
\[
\begin{tikzpicture}[scale=0.8]
\foreach \x in {1,2,3}
  \draw (-.5,-\x) node {Row \x};
  \foreach \x in {2,3,4,5,6,7,9}
  \draw (\x, -1) circle (.25);
\foreach \x in {1,8}
  \fill[black] (\x, -1) circle (.1);
\foreach \x in {1,2,4,8}
  \draw (\x, -2) circle (.25);
\foreach \x in {3,5,6,7,9}
  \fill[black] (\x, -2) circle (.1);
\foreach \x in {3,5}
  \draw (\x, -3) circle (.25);
  \foreach \x in {1,2,4,6,7,8,9}
  \fill[black] (\x, -3) circle (.1);
\end{tikzpicture}
\]

\end{example}

\subsection{Multiline queue}\label{subsec:mq}
Let us introduce a \defn{pairing} of a ball system.
Given a ball system
$\bb = \bb_n \otimes \cdots \otimes \bb_1 \in B(\mm)$,
we set $\bb^{(n)} = \bb, \mm^{(n)}= \mm$ and 
construct a family of smaller ball systems 
\begin{align}\label{a:bmr}
\bb^{(r)} \in B({\bf m}^{(r)}) = B_{l^{(r)}_r} \otimes  \cdots \otimes 
B_{l^{(r)}_1}\qquad  (l^{(r)}_i=m_i+m_{i+1}+\cdots + m_r)
\end{align}
 in the order $r=n-1, n-2,\ldots, 1$.
The component $B_{l^{(r)}_i}$ corresponds to the Row $i$ of the ball system $\bb^{(r)}$.
Note that $l^{(r)}_r=m_r$.
Let $d_1,\ldots, d_{m_r}$ be the balls of Row $r$ in the ball diagram of (\ref{a:bmr}) from left to right in this order.
We pair the ball $d_\alpha$ to a ball $d'_\alpha$ in Row $(r-1)$ in the order $\alpha=1,2,\ldots, m_r$.
The choice of $d'_\alpha$ is arbitrary except that if there is a yet unpaired ball just above $d_\alpha$,
that must be selected.  (Such a case is called trivial pairing.)
The partners $d'_1,\ldots, d'_{m_r}$ are not necessarily aligned from left to right 
in Row $(r-1)$.
Let $d''_1,\ldots, d''_{m_r}$ be these balls read from left to right.
They are paired to the balls in Row $(r-2)$ in this order similarly.
Repeating this process we get a sequence of injections 
\begin{align}\label{a:cr}
 \{ \text{balls in Row $r$ of $\bb^{(r)}$} \} \hookrightarrow 
 \{ \text{balls in Row $(r-1)$ of $\bb^{(r)}$} \}
 \hookrightarrow \cdots \hookrightarrow 
 \{ \text{balls in Row $1$ of $\bb^{(r)}$} \}.
 \end{align}
Denote the final image by $\cc_r \in B_{m_r}$.
This procedure will be referred to as the $(n+1-r)$ th round  of the whole pairing.
After completion of it, we assign a \defn{color} $r$ to all the balls captured in the round, and 
eliminate them. The resulting ball system defines $\bb^{(r-1)}$.
A pairing of $\bb \in B(\mm)$ is obtained by doing the round  $1, 2, \ldots, n$ in this order,
where the  last round $n$ actually does not introduce an injection but only 
assigns the color $1$ to the remaining balls in the Row 1.

A pairing of a ball system  $\bb \in B(\mm)$ may be regarded as a collection 
$\phi=(\phi_{1,2},\ldots, \phi_{n-1,n})$ of injections
\begin{align}\label{a:inj}
\phi_{s-1,s} \colon 
 \{ \text{balls in Row $s$ of $\bb$} \} \hookrightarrow \{ \text{balls in Row $(s-1)$ of $\bb$} \}
 \end{align}
satisfying a certain condition.
A ball system assigned with a pairing is called a \defn{multiline queue} (MLQ) and 
denoted by $Q = (\phi, \bb)$.
We identify it with a MLQ diagram, which is the ball diagram endowed with an arrow $d' \leftarrow d$ assigned to 
each pair of balls $d, d'$ such that $d'=\phi_{r-1,r}(d)$ for some 
$2 \le r \le n$. 
The arrow starts from $d$ and proceeds to the left, cyclically wrapping if necessary, until it 
reaches $d'$ upstairs. 
There are in general many MLQs $Q = (\phi,\bb)$ for a given ball system $\bb$. 
The set of MLQs built in this way based on the set of ball systems $B(\mm)$ in 
(\ref{a:bk}) will be denoted by MLQ$(\mm)$.

Each ball $d$ in $Q = (\phi,\bb) \in \text{MLQ}(\mm)$ is 
uniquely colored as $\mathcal{C}(d) \in \{1,2,\ldots, n\}$ by the rule explained in the above.
By construction, the colors of balls in Row $r$ range over $\{r, r+1, \ldots, n\}$.
In particular, there are exactly $m_\alpha$ balls of color $\alpha$ for $1 \le \alpha \le n$
in Row 1 of any MLQ from $\text{MLQ}(\mm)$.
We will always understand that balls in MLQs have been colored.
  
\begin{example}\label{a:ex2}
The following is a MLQ for the ball system in Example~\ref{a:ex1}.
Non-trivial pairings are labeled as $p_1,\ldots, p_4$ for later convenience.
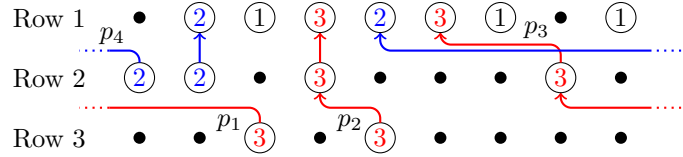
\[
\begin{tikzpicture}[scale=0.8]
\foreach \x in {1,2,3}
  \draw (-.5,-\x) node {Row \x};
\foreach \x in {3,5}
  \draw (\x, -3) circle (.25) node[color=red] {$3$};
\foreach \x in {1,2,4,6,7,8,9}
  \fill[black] (\x, -3) circle (.1);
\draw[rounded corners,thick,->,color=red] (5,-2.75) -- (5,-2.5) -- (4,-2.5) -- (4,-2.25);
\draw (4.5,-2.75) node {$p_2$};
\draw[rounded corners,thick,color=red] (3,-2.75) -- (3,-2.5) -- (.5,-2.5);
\draw (2.5,-2.75) node {$p_1$};
\draw[dotted,thick,color=red] (.5,-2.5) -- (0,-2.5);
\draw[dotted,thick,color=red] (10,-2.5) -- (9.5,-2.5);
\draw[rounded corners,thick,->,color=red] (9.5,-2.5) -- (8,-2.5) -- (8,-2.25);
\foreach \x in {1,2,4,8}
  \draw (\x, -2) circle (.25);
\draw[color=red] (8,-2) node {$3$};
\draw[color=red] (4,-2) node {$3$};
\draw[color=blue] (2,-2) node {$2$};
\draw[color=blue] (1,-2) node {$2$};
\foreach \x in {3,5,6,7,9}
  \fill[black] (\x, -2) circle (.1);
\draw[rounded corners,thick,->,color=red] (4,-1.75) -- (4,-1.25);
\draw[rounded corners,thick,->,color=red] (8,-1.75) -- (8,-1.45) -- (6,-1.45) -- (6,-1.25);
\draw (7.6,-1.2) node {$p_3$};
\draw[thick,->,color=blue] (2,-1.75) -- (2,-1.25);
\draw[rounded corners,thick,color=blue] (1,-1.75) -- (1,-1.55) -- (.5,-1.55);
\draw (0.55,-1.3) node {$p_4$};
\draw[dotted,thick,color=blue] (.5,-1.55) -- (0,-1.55);
\draw[dotted,thick,color=blue] (10,-1.55) -- (9.5,-1.55);
\draw[rounded corners,thick,->,color=blue] (9.5,-1.55) -- (5,-1.55) -- (5,-1.25);
\foreach \x in {2,3,4,5,6,7,9}
  \draw (\x, -1) circle (.25);
\draw[color=red] (4,-1) node {$3$};
\draw[color=red] (6,-1) node {$3$};
\draw[color=blue] (2,-1) node {$2$};
\draw[color=blue] (5,-1) node {$2$};
\draw (3,-1) node {$1$};
\draw (7,-1) node {$1$};
\draw (9,-1) node {$1$};
\foreach \x in {1,8}
  \fill[black] (\x, -1) circle (.1);
\end{tikzpicture}
\]
\end{example}

\subsection{ASEP state $\ket{P_{\text{MLQ}}(\mm)_q}$}
 
Here we construct an ASEP state  $\ket{P_{\text{MLQ}}(\mm)_q} \in W(\mm)$
in three steps.
It is known as the MLQ construction of the stationary state~\cite{CMW,Martin20}.
For ASEP, only the $q = 1$ case is necessary.
However, we explain a generalization including generic $q$ that was introduced in~\cite{CMW} for applications to Macdonald polynomials.
We shall focus on the construction process here.
The connection to the matrix product method in Section~\ref{a:subsec:mp} and a proof that 
$\ket{P_{\text{MLQ}}(\mm)_{q=1}}$ is indeed the stationary state will be presented in later sections.
\begin{description}
\item[Step 1]
We define a map from MLQs to ASEP configurations (see~\eqref{a:VP} for the 
definition of $\Sigma(\mm)$)
\begin{align}\label{a:pi}
\pi\colon \text{MLQ}(\mm) \rightarrow \Sigma(\mm); \qquad Q \mapsto (\sigma_1,\ldots, \sigma_L)
\end{align}
by stating that the image is the configuration of colored balls in Row 1 of $Q$, where 
empty boxes are regarded as 0.
Concretely, $\sigma_j = \mathcal{C}(d) \in \{1,\ldots, n\}$ if $d$ is 
the ball corresponding to $b_{1,j} = 1$ and 
$\sigma_j = 0$ if $b_{1,j} = 0$. 
For $Q$ in Example~\ref{a:ex2}, 
$\bb_1=(b_{1,1},\ldots, b_{1,8})$ 
is given in Example~\ref{a:ex1} and $\pi(Q) = (0,2,1,3,2,3,1,0,1)$.

\item[Step 2]
We assign a \defn{weight} ${\wt}_{q,t}$ to a MLQ $Q = (\phi, \bb)$ as
\begin{align}\label{a:wt1}
\wt_{q,t}(Q) = \prod_p \wt_{q,t}(p),
\end{align}
where the product is taken over all the pairs of balls $p = (d' \leftarrow d)$ in $Q$
specified by $\phi$ as $d' = \phi_{r-1,r}(d)$ for some $2 \le r \le n$.
Let $\mathcal{B}_{c,r}$ be the set of the balls in Row $r$ and Row $(r\!-\!1)$ having the color $c$.
There are $m_c$ such balls in both rows and they are paired by $\phi_{r-1,r}$.
We determine the weights in the following order according to the pairing procedure explained  in the previous subsection:
\begin{equation}\label{a:BB}
\begin{split}
&\mathcal{B}_{n,n}, \mathcal{B}_{n,n-1},\ldots, \mathcal{B}_{n,2},
\\
&\mathcal{B}_{n-1,n-1},\ldots, \mathcal{B}_{n-1,2},
\\
& \qquad \cdots \\
&\mathcal{B}_{2,2}.
\end{split}
\end{equation}
The $i$ th row here corresponds to the $i$ th round of the pairing.
Following the previous subsection,  
we consider the balls in Row $r$ from left to right within each $\mathcal{B}_{c,r}$.\footnote{
The sum of the resulting weights  (\ref{a:m1})  is actually independent of this order as shown in~\cite[Lem.2.1]{CMW}.}
When seeking a pairing partner of 
a ball $d$ in Row $r$, the balls $d'$ in Row $(r-1)$ that are not yet paired are called \defn{free}.
As mentioned after (\ref{a:inj}), the pairing $p = (d' \leftarrow d)$ is depicted as an arrow
going from $d$ to the left cyclically until it ends at $d'$ upstairs.
Suppose that $d$ and $d'$ are in the $j$-th and the $j'$-th columns of the 
MLQ diagram from the left.
If $j'=j$, it is a trivial pairing and we set $\wt_{q,t}(p)=1$.
Suppose $j'\neq j$.
Set $\delta_{\text{wrap}}=1$ if the arrow is wrapping, i.e.,  $j<j'$.
Otherwise we set $\delta_{\text{wrap}}=0$.
The free balls in Row $(r-1)$ in the columns $j'+1, j'+2,\ldots, j-1$
(indices regarded as elements in $ \Z_L$ here) are called \defn{skipped}.
Now the weight is given as
\begin{align}
\wt_{q,t}(p) =
\dfrac{(1-t)t^{\#\text{skipped}}q^{(c-r+1)\delta_{\text{wrap}}}}{1-q^{c-r+1}t^{\#\text{free}}}
=
 \begin{cases}
\dfrac{(1-t)t^{\#\text{skipped}}}{1-q^{c-r+1}t^{\#\text{free}}}
& \text{if $j'<j$}, \\[.7em]
\dfrac{(1-t)t^{\#\text{skipped}}q^{c-r+1}}{1-q^{c-r+1}t^{\#\text{free}}}
& \text{if $j<j'$}.
\end{cases}
\label{a:wt2}
\end{align}
In Example~\ref{a:ex2}, the weights of the non-trivial pairings are
\begin{equation}
\begin{split}
&\wt_{q,t}(p_1) = \frac{qt^2(1-t)}{1-qt^4},
\;\;
\wt_{q,t}(p_2) = \frac{1-t}{1-qt^3},
\;\;
\wt_{q,t}(p_3) = \frac{t(1-t)}{1-q^2t^6},
\;\;
\wt_{q,t}(p_4) = \frac{qt^2(1-t)}{1-qt^5},
\end{split}
\end{equation}
and $\wt_{q,t}(Q)$ is the product of them. 

\item[Step 3] By using $\pi$ (\ref{a:pi}) and the weight  (\ref{a:wt1}), 
the state $\ket{P_{\text{MLQ}}(\mm)_q}$ is constructed as 
\begin{equation}\label{Smlq}
\ket{P_{\text{MLQ}}(\mm)_q}
= \sum_{Q \in \text{MLQ}(\mm)}
\wt_{q,t}(Q)\ket{\pi(Q)}.
\end{equation}
\end{description}

We note that when $q=1$, the weight $\wt_{q,t}(Q)$ is invariant under (horizontal) $\Z_L$ cyclic shifts of $Q$.
Therefore $\ket{P_{\text{MLQ}}(\mm)_{q=1}}$ is translationally invariant.

\begin{example}
Consider the $(L,n) = (4,2)$ case.
The MLQs contributing to the states $|1012\rangle, |1021\rangle, |2011\rangle 
\in W(1,2,1)$ 
and their weights are given as
\[
% Lazy way to get the correctly oriented picture: rotate by 180 degrees!
\begin{array}{c@{\hspace{30pt}}c@{\hspace{20pt}}c@{\hspace{30pt}}c@{\hspace{20pt}}c}
\ket{1012}: &
\begin{tikzpicture}[scale=.8,baseline=-.5cm,rotate=180]
\draw (1, 1) circle (.25) node[color=red] {$2$};
\draw (1, 0) circle (.25) node[color=red] {$2$};
\draw (2, 0) circle (.25) node[color=blue] {$1$};
\draw (4, 0) circle (.25) node[color=blue] {$1$};
\foreach \x in {2,3,4}
  \fill[black] (\x, 1) circle (.1);
\fill[black] (3, 0) circle (.1);
\draw[rounded corners,thick,->,color=red] (1,.75) -- (1,.25);
\end{tikzpicture}
&
1,
&
\begin{tikzpicture}[scale=.8,baseline=-.5cm,rotate=180]
\draw (3, 1) circle (.25) node[color=red] {$2$};
\draw (1, 0) circle (.25) node[color=red] {$2$};
\draw (2, 0) circle (.25) node[color=blue] {$1$};
\draw (4, 0) circle (.25) node[color=blue] {$1$};
\foreach \x in {1,2,4}
  \fill[black] (\x, 1) circle (.1);
\fill[black] (3, 0) circle (.1);
\draw[rounded corners,thick,->,color=red] (3,.75) -- (3,.5) -- (4.5,.5);
\draw[dotted,thick,color=red] (4.5,.5) -- (5,.5);
\draw[dotted,thick,color=red] (0,.5) -- (.5,.5);
\draw[rounded corners,thick,->,color=red] (.5,.5) -- (1,.5) -- (1,.25);
\end{tikzpicture}
&
\dfrac{qt(1-t)}{1-qt^3};
\allowdisplaybreaks \\[1cm]
\ket{1021}: &
\begin{tikzpicture}[scale=.8,baseline=-.5cm,rotate=180]
\draw (2, 1) circle (.25) node[color=red] {$2$};
\draw (2, 0) circle (.25) node[color=red] {$2$};
\draw (1, 0) circle (.25) node[color=blue] {$1$};
\draw (4, 0) circle (.25) node[color=blue] {$1$};
\foreach \x in {1,3,4}
  \fill[black] (\x, 1) circle (.1);
\fill[black] (3, 0) circle (.1);
\draw[rounded corners,thick,->,color=red] (2,.75) -- (2,.25);
\end{tikzpicture}
&
1,
&
\begin{tikzpicture}[scale=.8,baseline=-.5cm,rotate=180]
\draw (3, 1) circle (.25) node[color=red] {$2$};
\draw (2, 0) circle (.25) node[color=red] {$2$};
\draw (1, 0) circle (.25) node[color=blue] {$1$};
\draw (4, 0) circle (.25) node[color=blue] {$1$};
\foreach \x in {1,2,4}
  \fill[black] (\x, 1) circle (.1);
\fill[black] (3, 0) circle (.1);
\draw[rounded corners,thick,->,color=red] (3,.75) -- (3,.5) -- (4.5,.5);
\draw[dotted,thick,color=red] (4.5,.5) -- (5,.5);
\draw[dotted,thick,color=red] (0,.5) -- (.5,.5);
\draw[rounded corners,thick,->,color=red] (.5,.5) -- (2,.5) -- (2,.25);
\end{tikzpicture}
&
\dfrac{qt^2(1-t)}{1-qt^3};
\allowdisplaybreaks \\[1cm]
\ket{2011}: &
\begin{tikzpicture}[scale=.8,baseline=-.5cm,rotate=180]
\draw (4, 1) circle (.25) node[color=red] {$2$};
\draw (4, 0) circle (.25) node[color=red] {$2$};
\draw (1, 0) circle (.25) node[color=blue] {$1$};
\draw (2, 0) circle (.25) node[color=blue] {$1$};
\foreach \x in {1,2,3}
  \fill[black] (\x, 1) circle (.1);
\fill[black] (3, 0) circle (.1);
\draw[rounded corners,thick,->,color=red] (4,.75) -- (4,.25);
\end{tikzpicture}
&
1,
&
\begin{tikzpicture}[scale=.8,baseline=-.5cm,rotate=180]
\draw (3, 1) circle (.25) node[color=red] {$2$};
\draw (4, 0) circle (.25) node[color=red] {$2$};
\draw (1, 0) circle (.25) node[color=blue] {$1$};
\draw (2, 0) circle (.25) node[color=blue] {$1$};
\foreach \x in {1,2,4}
  \fill[black] (\x, 1) circle (.1);
\fill[black] (3, 0) circle (.1);
\draw[rounded corners,thick,->,color=red] (3,.75) -- (3,.5) -- (4,.5) -- (4,.25);
\end{tikzpicture}
&
\dfrac{1-t}{1-qt^3}.
\end{array}
\]
Summing them at $q=1$ yields
\begin{align}
(1+t+t^2) \ket{P_{\text{MLQ}}(1,2,1)_{q=1}}
= (1+t)^2 \ket{1012} + (1+t+2t^2) \ket{1021} + (2+t+t^2) \ket{2011} + \cdots,
\end{align}
which agrees with $\ket{\xi(1,2,1)}$ in Example~\ref{ex:pbar} in view of the cyclic $\Z_L$ symmetry. 
\end{example}

\subsection{Matrix $\widecheck{M}$}

We introduce a matrix $\widecheck{M}$ 
which describes the interaction between neighboring rows in MLQs. 
Consider a MLQ for $n=2$ case with $(l_1,l_2)=(m,l)\, (l<m)$.
It contains two rows and have the form
$Q=(\phi,{\bf i} \otimes {\bf j})$ with ${\bf i} \otimes {\bf j} \in B_l \otimes B_m$.
Here ${\bf i} =(i_1,\ldots, i_L) \in B_l$ 
means that there is $i_k(=0,1)$ ball in the $k$~th column from the left in the lower Row 2.
Similarly ${\bf j} = (j_1,\ldots, j_L) \in B_m$ specifies the positions of balls in the upper Row 1.

For any 
$\aaa \otimes \bb\in B_l \otimes B_{m-l}$ and 
${\bf i} \otimes {\bf j} \in B_l \otimes B_m$ with $l<m$, define 
a generating function of the weights by
\begin{equation}
M(q,t)^{\aaa, \bb}_{{\bf i}, {\bf j}}
= \delta^{\aaa+\bb}_{\bf j}\sum_{\phi} \text{wt}_{q,t}(Q),
\label{a:m1}
\end{equation}
where $\delta^{\bf x}_{\bf y} = \theta({\bf x}={\bf y})$.
The sum is taken over the pairings 
$\phi$ satisfying the condition
\begin{align}
\phi(\{\text{balls in ${\bf i}$}\}) = \{\text{balls in $\aaa$}\}.
\label{a:m2}
\end{align}
To summarize, the indices $\aaa, \bb,\ii,\jj$ of $M(q,t)^{\aaa, \bb}_{{\bf i}, {\bf j}}$ have the following meaning,
where the last column is an interpretation in the language of the queuing processes in~\cite{Martin20}:
\begin{equation}
\begin{array}{c|l|l}
\text{Indices} & \text{ball picture} & \text{queuing process} \\
\hline
B_l \ni \mathbf{i} & \text{balls in the lower row} & \text{arrival} \\
B_m \ni \mathbf{j} & \text{balls in the upper row} & \text{service} \\
B_l \ni \mathbf{a} & \text{paired balls in the upper row} & \text{departure} \\
B_{m-l} \ni \mathbf{b} & \text{unpaired balls } \mathbf{j}-\mathbf{a} \text{ in the upper row} & \text{unused service} \\
\end{array}
\label{a:abij}
\end{equation}

\noindent
The constraint $\aaa + \bb={\bf j}$ in  (\ref{a:m1}) is natural from the queuing process interpretation.
In general, there are numerous choices for $\phi$ contributing to the sum (\ref{a:m1}), with a maximum of $l!$.

\begin{example}\label{a:ex:M}
Consider the two MLQs for $(l_1,l_2)=(\alpha+\beta+1, 2)$ as follows:
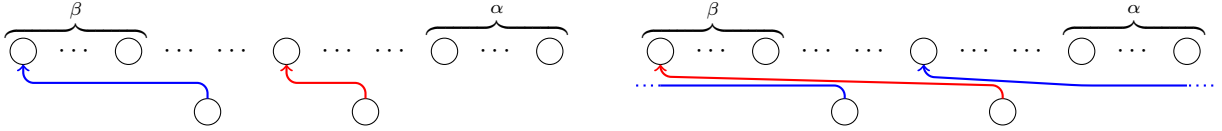
\[
% Lazy way to get the correctly oriented picture: rotate by 180 degrees!
\begin{tikzpicture}[scale=0.7,rotate=180]
\foreach \x in {1,3,6,9,11}
  \draw (\x, 0) circle (.25);
\foreach \x in {2,4,5,7,8,10}
  \draw (\x, 0) node {$\cdots$};
\draw (4.5, 1.15) circle (.25);
\draw (7.5, 1.15) circle (.25);
\draw[rounded corners,thick,->,color=red] (4.5,.9) -- (4.5,.6) -- (6,.6) -- (6,.25);
\draw[rounded corners,thick,->,color=blue] (7.5,.9) -- (7.5,.6) -- (11,.6) -- (11,.25);
\draw (2,-.6) node {$\overbrace{\hspace{55pt}}^{\alpha}$};
\draw (10,-.6) node {$\overbrace{\hspace{55pt}}^{\beta}$};
\end{tikzpicture}
\qquad
\begin{tikzpicture}[scale=0.7,rotate=180]
\foreach \x in {1,3,6,9,11}
  \draw (\x, 0) circle (.25);
\foreach \x in {2,4,5,7,8,10}
  \draw (\x, 0) node {$\cdots$};
\draw (4.5, 1.15) circle (.25);
\draw (7.5, 1.15) circle (.25);
\draw[rounded corners,thick,->,color=red] (4.5,.9) -- (4.5,.65) -- (11,.48) -- (11,.25);
\draw[rounded corners,thick,color=blue] (7.5,.9) -- (7.5,.65) -- (11,.65);
\draw[dotted,thick,color=blue] (11,.65) -- (11.5,.65);
\draw[dotted,thick,color=blue] (.5,.65) -- (1,.65);
\draw[rounded corners,thick,->,color=blue] (1,.65) -- (3,.65)-- (6,.48) -- (6,.25);
\draw (2,-.6) node {$\overbrace{\hspace{55pt}}^{\alpha}$};
\draw (10,-.6) node {$\overbrace{\hspace{55pt}}^{\beta}$};
\end{tikzpicture}
\]
For the corresponding $\aaa, \bb, {\bf i}, {\bf j}$ one has 
\begin{align}\label{a:Mex}
M(q,t)^{\aaa, \bb}_{{\bf i}, {\bf j}} 
= \frac{t^{\beta-1}(1-t)}{1-qt^{\alpha+\beta+1}}
\frac{1-t}{1-qt^{\alpha+\beta}}
+
\frac{qt^{\alpha+\beta}(1-t)}{1-qt^{\alpha+\beta+1}}
\frac{t^{\beta-1}(1-t)}{1-qt^{\alpha+\beta}}
=
\frac{t^{\beta-1}(1-t)^2(1+q t^{\alpha+\beta})}
{(1-qt^{\alpha+\beta})(1-qt^{\alpha+\beta+1})}.
\end{align}
\end{example}

For $l \in  \Z_{\ge 1}$, let $V_l$ be the vector space having a basis $\{v_\bb\}$ labeled by $B_l$ from~\ref{a:B}:
\begin{align}
V_l = \bigoplus_{\bb \in B_l}\mathbb{C}(q,t)v_\bb.
\end{align}
For $l<m$,  we define a linear operator 
$\widecheck{M}(z,t)$ depending on $t$ and another variable $z$ by
\begin{subequations}
\label{eq:MVV}
\begin{align}
\widecheck{M}(z,t)\colon V_l \otimes V_m & \rightarrow V_{m-l} \otimes V_l
\label{a:Mvv0} \\
v_\ii \otimes v_\jj \;&\mapsto
\sum_{\aaa \otimes \bb \in B_l \otimes B_{m-l}}
M(z,t)^{\aaa, \bb}_{\ii, \jj} \,v_\bb\otimes v_\aaa
\qquad 
(\ii \otimes \jj \in B_l \otimes B_m),
\label{a:Mvv}
\end{align}
\end{subequations}
where the double sum in the RHS is actually the single one 
$\sum_{\aaa  \in B_l, \aaa \le \jj} M(z,t)^{\aaa, \jj-\aaa}_{\ii, \jj} \,v_{\jj-\aaa}\otimes v_\aaa$
since $M(z,t)^{\aaa, \bb}_{\ii, \jj}=0$ unless $\bb=\jj-\aaa$ by (\ref{a:m1}).

\begin{remark}
One might think that by introducing $M'(q,t)^{\aaa, \bb}_{{\bf i}, {\bf j}} := \delta^{\aaa+\bb}_{{\bf i} + {\bf j}}\sum_{\phi} \text{wt}_{q,t}(Q)$,  setting
\begin{align*}
\widecheck{M}'(z,t) \colon V_l \otimes V_m & \rightarrow V_m \otimes V_l
\\
v_{\ii} \otimes v_{\jj} & \mapsto \sum_{\aaa \otimes \bb \in B_l \otimes B_m} 
M'(z,t)^{\aaa, \bb}_{\ii, \jj} \,v_\bb\otimes v_\aaa
\end{align*}
is more natural rather than (\ref{a:m1}) and (\ref{eq:MVV}) since it possesses the standard ``weight conservation''' property common in quantum $R$ matrices.
The reason we employ the strange $\widecheck{M}(z,t)$ in (\ref{eq:MVV}) 
is to make it fit with the queuing process interpretation in (\ref{a:abij}) 
and will further be detailed in the next subsection.
We additionally note that the Yang--Baxter equation 
\[
\widecheck{M}'(x,t)_{1,2}\widecheck{M}'(x y,t)_{2,3}\widecheck{M}'(y,t)_{1,2} 
= \widecheck{M}'(y,t)_{2,3}\widecheck{M}'(x y,t)_{1,2}\widecheck{M}'(x,t)_{2,3}
\]
is not valid for generic $x$ and $y$.
\end{remark}

\subsection{ASEP state $|P_{\text{MLQ}}({\bf m})_q\rangle$ from $\widecheck{M}$}

Let us depict~\eqref{a:Mvv} in a conventional diagram for vertex models (see, \textit{e.g.},~\cite{Bax}):

\begin{equation}
\label{t:vm}
\begin{tikzpicture}[scale=.6,>=latex,baseline=0]
\draw[->,line width=.15em] (-1,0) node[anchor=east] {$\ii$} -- (1,0) node[anchor=west] {$\aaa$};
\draw[->,line width=.15em] (0,-1) node[anchor=north] {$\jj$} -- (0,1) node[anchor=south] {$\bb$};
\draw[-] (-.3, -0.1) arc (180:270:.2) node[anchor=north east,scale=0.9] {$z$};
\end{tikzpicture}
\quad
\longleftrightarrow
\quad
M(z,t)^{\aaa, \bb}_{\ii, \jj},
\end{equation}
where we are taking $\bb = \jj - \aaa$.
The arrows here and in the rest of this section, as seen in (\ref{eq:f7}) 
and (\ref{eq:f8}), correspond to the indices from $B_l$ in (\ref{a:B}). 
In contrast, the arrows in the next section, except those in (\ref{a:bd0})--(\ref{a:bd2}), 
carry a single integer.  In order to distinguish them, 
we use thick arrows for the former and thin arrows for the latter.

In what follows $\widecheck{M}(z,t)$ will simply be denoted by $\widecheck{M}(z)$ as the parameter $t$ is fixed everywhere.
Let $\widecheck{M}(z)_{j,j-1}$ be the operator acting on the $(j,j\!-\!1)$~th components 
$V_{s_j} \otimes V_{s_{j-1}}$ in $V_{s_n} \otimes \cdots \otimes V_{s_1}$ as $\widecheck{M}(z)$ and as the identity elsewhere ($s_{\alpha}$'s are arbitrary positive integers).
 
A key role in our work is played by the operator 
 \begin{align}
 \mathbb{M}(q,t) &: V_{l_n} \otimes V_{l_{n-1}} \otimes \cdots  \otimes V_{l_1} 
 \rightarrow V_{m_1} \otimes V_{m_2} \otimes \cdots \otimes V_{m_n},
 \end{align}
 where $m_i$'s and $l_i$'s are related by~\eqref{a:lk}.
 It is given as a composition of $\frac{1}{2}n(n-1)$ $\widecheck{M}$'s as
 \begin{align}
 \mathbb{M}(q,t) &=  A_{n-1}A_{n-2}\cdots A_1,
 \\
 A_j &= \widecheck{M}(q^{n-j})_{j+1,j} \widecheck{M}(q^{n-j-1})_{j+2,j+1} \cdots \widecheck{M}(q)_{n,n-1}.
 \label{a:A}
 \end{align}
Explicitly, it reads
 \begin{equation}\label{a:mod}
 \begin{split}
 \mathbb{M}(q,t)=\qquad \qquad \qquad \qquad \qquad \qquad \quad 
  \widecheck{M}(q)_{n,n-1}& \\
 \times  \widecheck{M}(q^2)_{n-1,n-2} \widecheck{M}(q)_{n,n-1} & \\
 \cdots \qquad \\
 \times \widecheck{M}(q^{n-1})_{2,1}  \widecheck{M}(q^{n-2})_{3,2} 
 \cdots \widecheck{M}(q)_{n,n-1},
 \end{split}
 \end{equation} 
 where the $r$th row from the bottom is $A_r$.
 Let us illustrate the $n=3$ case:
\begin{align}
 \mathbb{M}(q,t) &=
 \underbrace{\widecheck{M}(q)_{3,2}}_{A_2}\,
\underbrace{\widecheck{M}(q^2)_{2,1}\widecheck{M}(q)_{3,2}}_{A_1}.
\end{align}
 The matrix elements of $A_1$ for the transition
 $v_{\bb_3} \otimes v_{\bb_2} \otimes v_{\bb_1} \rightarrow
v_{\bb'_2} \otimes v_{\bb'_1} \otimes  v_{{\bf c}_3}$ 
and those of $A_2$ for 
$v_{\bb'_2} \otimes v_{\bb'_1} \otimes v_{{\bf c}_3}  \rightarrow 
v_{{\bf c}_1} \otimes v_{{\bf c}_2} \otimes v_{{\bf c}_3}$ are depicted as 
\begin{equation}
\begin{aligned}
\begin{tikzpicture}[scale=0.9,>=latex]
\foreach \x in {1,2}
  \draw[->,line width=.15em] (-\x,0.3) node[anchor=north] {$\bb_{\x}$} -- (-\x,1.7) node[anchor=south] {$\bb'_{\x}$};
\draw[->,line width=.15em] (-2.7,1) node[anchor=east] {$\bb_3$} -- (-0.3,1) node[anchor=west] {$\cc_3$};
\draw[-] (-2-.2, 1) arc (180:270:.2) node[anchor=north east,scale=.8] {$q$};
\draw[-] (-1-.2, 1) arc (180:270:.2) node[anchor=north east,scale=.8] {$q^2$};
\end{tikzpicture}
\qquad \qquad 
\begin{tikzpicture}[scale=0.9,>=latex]
\draw[->,line width=.15em] (-1,0.3) node[anchor=north] {$\bb'_1$} -- (-1,1.7) node[anchor=south] {$\cc_1$};
\draw[->,line width=.15em] (-1.7,1) node[anchor=east] {$\bb'_2$} -- (-0.3,1) node[anchor=west] {$\cc_2$};
\draw[-] (-1-.2, 1) arc (180:270:.2) node[anchor=north east,scale=.8] {$q$};
\end{tikzpicture}
\end{aligned}
\label{eq:f7}
\end{equation}
Let us explain the meaning of these diagrams along with Example~\ref{a:ex2}. 
The left diagram in (\ref{eq:f7}) for $A_1$ 
shows the first round of the pairing process corresponding to the red arrows in Example~\ref{a:ex2}. 
One lets $\bb_3$ on Row 3 ``penetrate" Row 2 and then Row 1, obtaining the image ${\bf c}_3$ 
which specifies the location of color 3 balls at the top. See the definition of ${\bf c}_r$ given after (\ref{a:cr}).
The elements $\bb'_2$ and $\bb'_1$ represent the free balls left intact in Row 2 and Row 1 
in the first round, respectively.
In the second round of the pairing, one lets $\bb'_2$ interact with $\bb'_1$ 
as indicated by the blue arrows in Example~\ref{a:ex2}. 
This is depicted in the right diagram in (\ref{eq:f7}) for $A_2$, 
where ${\bf c}_2$ and ${\bf c}_1$ correspond to the color 2 and 1 balls in Row 1, respectively. 
In this way, ${\bf c}_1, {\bf c}_2, {\bf c}_3$ give the final list of color 1, 2, and 3 balls in Row 1. 
The weight $\text{wt}(Q)_{q,t}$ of the MLQ is equal to the element of $\mathbb{M}(q,t)$ for the transition 
$v_{\bb_3} \otimes v_{\bb_2} \otimes v_{\bb_1} 
\rightarrow v_{{\bf c}_1} \otimes v_{{\bf c}_2} \otimes v_{{\bf c}_3}$.
 Its diagram is obtained by combining the two in (\ref{eq:f7}).
 It results in a \textit{single} diagram in the left of the following:
 
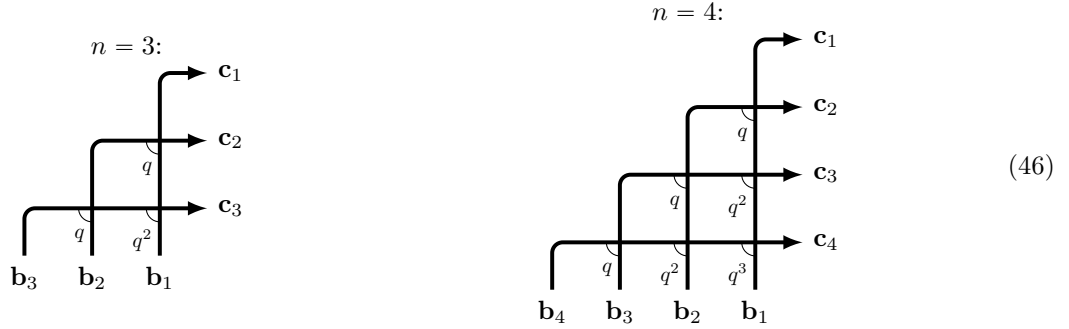
\begin{equation}
\begin{minipage}{0.45\textwidth}
\centering
\text{$n=3$:} \\
\begin{tikzpicture}[scale=0.9,>=latex]
\foreach \x in {1,2,3}
  \draw[->, rounded corners, line width=.15em] (-\x,0.3) node[anchor=north] {$\mathbf{b}_{\x}$} -- (-\x,4-\x) -- (-0.3,4-\x) node[anchor=west] {$\mathbf{c}_{\x}$};
\foreach \x in {1,2}
  \draw[-] (\x-3-.2, \x) arc (180:270:.2) node[anchor=north east,scale=.8] {$q$};
\draw[-] (-1-.2, 1) arc (180:270:.2) node[anchor=north east,scale=.8] {$q^2$};
\end{tikzpicture}
\end{minipage}
\hfill
\begin{minipage}{0.45\textwidth}
\centering
\text{$n=4$:} \\
\begin{tikzpicture}[scale=0.9,>=latex]
\foreach \x in {1,2,3,4}
  \draw[->, rounded corners, line width=.15em] (-\x,0.3) node[anchor=north] {$\mathbf{b}_{\x}$} -- (-\x,5-\x) -- (-0.3,5-\x) node[anchor=west] {$\mathbf{c}_{\x}$};
\foreach \x in {1,2,3}
  \draw[-] (\x-4-.2, \x) arc (180:270:.2) node[anchor=north east,scale=.8] {$q$};
\foreach \x in {1,2}
  \draw[-] (\x-3-.2, \x) arc (180:270:.2) node[anchor=north east,scale=.8] {$q^2$};
\draw[-] (-1-.2, 1) arc (180:270:.2) node[anchor=north east,scale=.8] {$q^3$};
\end{tikzpicture}
\end{minipage}
\label{eq:f8}
\end{equation}

 \noindent
Note that $v_{\cc_1} \otimes v_{\cc_2} \otimes v_{\cc_3} 
 \in V_{m_1} \otimes V_{m_2} \otimes V_{m_3}$ from \eqref{a:Mvv0} and (\ref{a:lk}).
The right one in (\ref{eq:f8}) is the $n=4$ case.
 From (\ref{eq:f8}), the general $n$ case is clear. 
 The diagram for $\mathbb{M}(q,t)$ has the form of the \defn{corner transfer matrix} (CTM) 
 of the NW quadrant in~\cite[Fig.~13.2]{Bax}.\footnote{It coincides with a wiring diagram of the longest element 
 of the symmetric group $\mathfrak{S}_n$.}
The operator $A_j$ (\ref{a:A}) corresponds to the $j$~th row from the bottom in the diagram.
It encodes the weights in the $j$~th round of the pairing process, which concerns the balls in the $j$~th line of~\eqref{a:BB} from the top, and  they are colored in $n+1-j$.
The reason we employ the unusual $\bb = \jj - \aaa$ weight conservation in~\eqref{a:Mvv} and~\eqref{t:vm} 
is that it is necessary to describe the unpaired balls, which remain active and will be the relevant players in the 
subsequent rounds. See~\eqref{a:abij}.

The outputs of $\mathbb{M}(q,t)$ are superpositions of data of the form $v_{\cc_1}\otimes \cdots \otimes v_{\cc_n} \in V_{m_1} \otimes \cdots \otimes V_{m_n}$.
They are transformed to the ASEP states in $W(\mm)$~\eqref{a:VP} by a simple ``projection'':
\begin{equation}\label{a:Pi}
 \begin{aligned}
 \Pi \colon V_{m_1} \otimes \cdots \otimes V_{m_n} &\rightarrow 
 W(\mm)
 \\
 v_{\cc_1} \otimes \cdots \otimes v_{\cc_n} & \mapsto
 \ket{\cc_1+2\cc_2+\cdots + n\cc_n},
 \end{aligned}
\end{equation}
where the $0$~th component of $\mm = (m_0, m_1, \ldots, m_n)$ is determined by the condition $\abs{\mm} = L$.
To summarize the argument so far, we have explained that the state 
$\ket{P_{\text{MLQ}}(\mm)_q}$ constructed from the 
MLQ approach is expressed as follows.
 
 \begin{proposition}[CTM interpretation of the MLQ construction]\label{a:pr:P}
 \begin{align}
 \ket{P_{\text{MLQ}}(\mm)_q} = 
 \Pi\Bigl(
 \mathbb{M}(q,t)\sum_{\bb_n \otimes \cdots \otimes \bb_1 \in B(\mm)}
 v_{\bb_n} \otimes \cdots \otimes v_{\bb_1} 
 \Bigr).
 \end{align}
 \end{proposition}
  It is known to yield the actual stationary states at $q=1$~\cite{Martin20}.

\section{$t$-oscillator weighted five vertex model}\label{sec:5v}

We will compute the stationary probabilities by introducing a five vertex model.
It is different from those considered in~\cite{KMO1, KMO2} and 
does not satisfy the usual weight conservation property.
Instead, it uses the same strange weight conservation as in the previous section.

\subsection{$t$-deformed quantum oscillator}
Our five vertex model will be weighted using a $t$-oscillator algebra $\qoa$ 
generated by $\langle \aaa^+, \aaa^-, \kk \rangle$ satisfying the relations
\begin{align}\label{a:aak}
\kk \,\aaa^{\pm} = t^{\pm 1} \aaa^{\pm} \kk,
\qquad
\aaa^- \aaa^+ = 1 - t \kk,
\qquad
\aaa^+ \aaa^- = 1 - \kk.
\end{align}
It has a natural representation from its triangular decomposition, which we refer to as the \defn{bosonic Fock space}:
\begin{align}\label{a:F}
F := \bigoplus_{d=0}^{\infty} \Q(t) \ket{d};
\qquad
\kk \ket{d} = t^d \ket{d},
\quad
\aaa^+ \ket{d} = \ket{d+1},
\quad
\aaa^- \ket{d} = (1 - t^d) \ket{d-1},
\end{align}
extended by linearity with the convention $\ket{-1} := 0$.
We will also use the ``number" operator $\bf{h}$ defined by 
\begin{equation}\label{a:h}
\hh|d\rangle = d | d\rangle
\end{equation}
so that ${\bf k}= t^\hh$.

\begin{remark}\label{a:re:tau}
Consider a $\tau$-twist quantum oscillator algebra  
$\qoa_{\tau}$ having the relations
\begin{align}\label{a:aakt}
\kk \,\aaa^{\pm} = t^{\pm 1} \aaa^{\pm} \kk,
\qquad
\aaa^- \aaa^+ = 1 - t^{\tau} \kk^{\tau},
\qquad
\aaa^+ \aaa^- = 1 - \kk^{\tau}
\end{align}
and the representation of $F$:
\begin{align}\label{a:Ft}
\kk \ket{d} = t^d \ket{d},
\qquad
\aaa^+ \ket{d} = \ket{d+1},
\qquad
\aaa^- \ket{d} = (1 - t^{\tau d}) \ket{d-1},
\end{align}
in place of~\eqref{a:aak} and~\eqref{a:F}.
In this paper, we will use $\qoa=\qoa_{\tau = 1}$ only, but let us make a few remarks about the general case.
The representations of the quantum oscillator algebra can be used 
to construct representations of a (Drinfel'd--Jimbo) quantum group (see, \textit{e.g.},~\cite{Hayashi90,K22}), 
where the twist parameter $\tau$ is related to the deformation parameter in parallel to~\cite{AS24}.
When $\tau = 2$, this is the quantum oscillator algebra that has appeared in relation to the tetrahedron equation; see, \textit{e.g.},~\cite{KMO2,K22}.
The idea of realizing Lie algebra generators using bosons, i.e., 
creation and annihilation operators, in physics 
dates back to~\cite{HP}, which is referred to as the Holstein--Primakoff transformation.
\end{remark}

\subsection{$t$-oscillator weighted five vertex model}\label{subsec:5v}
We define our vertex model weights $S_{ij}^{ab}$ as follows:
\begin{equation}
\newcommand{\fvm}[4]{
\begin{tikzpicture}[scale=.35,>=latex]
\draw[->] (-1,0) node[anchor=east]{$#1$} -- (1,0) node[anchor=west]{$#2$};
\draw[->] (0,-1) node[anchor=north]{$#3$} -- (0,1) node[anchor=south]{$#4$};
\end{tikzpicture}
}%
\begin{array}{c*{5}{@{\hspace{20pt}}c}}
\fvm{i}{a}{j}{b}\quad &\fvm{0}{0}{0}{0} & \fvm{1}{1}{1}{0} & \fvm{0}{0}{1}{1} & \fvm{0}{1}{1}{0} & \fvm{1}{0}{0}{0} \\
S_{ij}^{ab} & 1 & 1 & \kk & \aaa^- & \aaa^+
\end{array}
\label{t:s5V}
\end{equation}
We set $S_{ij}^{ab}=0$ for all the other configurations.
Note that ${\bf k}$ corresponds to the unique configuration such that $j-a>0$.
The weight conservation holds in an unusual form:
\begin{align}\label{a:S0}
 S^{ab}_{ij}=0 \;\;\text{unless}\;\; \; a+b=j.
\end{align}
Compare~\eqref{a:S0} with~\eqref{a:m1}.
We let the number operator ${\bf h}$ also act on the $\{0,1\}$ states.
Then (\ref{t:s5V}) satisfies
\begin{equation}
\label{a:c23}
z^{{\bf h}+a}S^{ab}_{ij} = S^{ab}_{ij} z^{{\bf h}+i}, 
\end{equation} 
which implies that the total weight of the first and the $t$-oscillator components is conserved.

\begin{remark}\label{a:re:st}
The multispecies totally asymmetric simple exclusion process (TASEP) is the special case 
 corresponding to $t=0$ in the ASEP.
The five vertex model utilized in~\cite[Eq.(2.20)]{KMO1} 
for TASEP has the weights $\widetilde{S}^{ab}_{ij}|_{t=0}$,
where $\widetilde{S}^{ab}_{ij}$ is given by 
\newcommand{\fvm}[4]{
\begin{tikzpicture}[scale=.35,>=latex]
\draw[->] (-1,0) node[anchor=east]{$#1$} -- (1,0) node[anchor=west]{$#2$};
\draw[->] (0,-1) node[anchor=north]{$#3$} -- (0,1) node[anchor=south]{$#4$};
\end{tikzpicture}
}
\begin{equation}\label{a:stm}
\begin{array}{c*{5}{@{\hspace{20pt}}c}}
\fvm{i}{a}{j}{b}\quad &\fvm{0}{0}{0}{0} & \fvm{1}{1}{1}{1} & \fvm{0}{0}{1}{1} & \fvm{0}{1}{1}{0} & \fvm{1}{0}{0}{1}
\\
\widetilde{S}_{ij}^{ab} & 1 & 1 &  \kk & \aaa^- & \aaa^+
\end{array}
\end{equation}
It satisfies the usual weight conservation, i.e.,  $\widetilde{S}_{ij}^{ab}=0$ unless $a+b=i+j$.
The two models are related by $S^{ab}_{ij} = \widetilde{S}^{a,b+i}_{ij}$.
Both can be interpreted as three-dimensional vertex models in which the $t$-oscillator acts in the third direction.
The model (\ref{a:stm}) works efficiently for TASEP 
revealing the crystal theoretical nature~\cite{KMO1} of the MLQ method~\cite{FM}
and reducing the relevant ZF algebra to the tetrahedron equation at $t=0$~\cite{KMO2, K22}.
Although the ZF algebras for TASEP and ASEP are smoothly connected via the parameter $t$, 
we have not found a method to formulate the results in the present  paper based on the model 
(\ref{a:stm}).
\end{remark}

Next, for any 
$\aaa \otimes \bb\in B_l \otimes B_{m-l}$ and 
${\bf i} \otimes {\bf j} \in B_l \otimes B_m$ with $l<m$, we define 
$S(q,t)^{\aaa, \bb}_{\ii, \jj} $ 
 by a matrix product formula
\begin{equation}
\label{t:Sop}
S(q,t) _{\ii,\jj}^{\aaa,\bb}= (1 - q t^{m-\ell}) 
\tr_F \bigl(q^\hh S_{i_1j_1}^{a_1b_1} S_{i_2j_2}^{a_2b_2} \cdots S_{i_Lj_L}^{a_Lb_L} \bigr),
\end{equation}
where the trace is taken over $F$.
When the trace space is clear, we will omit it from the notation.
From (\ref{a:S0}) it follows that 
\begin{align}\label{a:S1}
S(q,t) _{\ii,\jj}^{\aaa,\bb}=0\;\; \text{ unless}\;\; \aaa + \bb = \jj.
\end{align}
Moreover from $|\jj |-|\aaa |=m-l>0$ and the comment after (\ref{t:s5V}), 
there is at least one ${\bf k}$ in the operators 
$S_{i_1 j_1}^{a_1 b_1}, \ldots, S_{i_L j_L}^{a_L b_L}$.
Therefore the trace (\ref{t:Sop}) is a valid formal power series in $t$ even at $q=1$.
The definition (\ref{t:Sop}) is depicted as~\cite[Fig.11.3]{K22}$|_{z\rightarrow q}$
if the vertices therein are regarded as $S^{a_r b_r}_{i_r j_r}$.
See also (\ref{a:fig9}) below.

\begin{example}\label{a:ex:S}
Suppose  $(l,m)=(2,\alpha+\beta+1)$ with $\alpha, \beta \ge 1$ and take
$\aaa, {\bf i} \in B_l, {\bf j} \in B_m$ and $\bb \in B_{m-l}$ as
\[
\aaa = (1,0^{\beta-1},0,1,0,0^\alpha),\;
\bb = (0,1^{\beta-1},0,0,0,1^\alpha),\;
\ii = (0,0^{\beta-1},1,0,1,0^\alpha),\;
\jj = (1,1^{\beta-1},0,1,0,1^\alpha),
\]
which satisfies $\aaa + \bb={\bf j}$. 
Now (\ref{t:Sop}) is calculated as
\begin{align*}
S(q,t) _{\ii,\jj}^{\aaa,\bb}&= (1 - q t^{\alpha+\beta-1}) 
\tr \bigl(q^\hh S^{10}_{01}(S^{01}_{01})^{\beta-1} S^{00}_{10}
S^{10}_{01}S^{00}_{10}(S^{01}_{01})^\alpha \bigr)
\\
&= (1 - q t^{\alpha+\beta-1}) 
\tr \bigl(q^\hh \aaa^- \kk^{\beta-1} \aaa^+ \aaa^- \aaa^+ \kk^\alpha \bigr)
\\
&=(1 - q t^{\alpha+\beta-1}) \sum_{d\ge 0}q^d(1-t^{d+1})t^{(\beta-1)(d+1)}(1-t^{d+1})t^{\alpha d}
= \frac{t^{\beta-1}(1-t)^2(1+qt^{\alpha+\beta})}{(1-qt^{\alpha+\beta})(1-qt^{\alpha+\beta+1})}.
\end{align*}
Note that this coincides with  $M(q,t) _{\ii,\jj}^{\aaa,\bb}$ in (\ref{a:Mex}).
The equality holds in general as shown in Theorem~\ref{t:th:ms} below, which is 
an essential result quantifying the algorithmic MLQ construction into the $t$-oscillator algebra.
\end{example}

\begin{theorem}\label{t:th:ms}
For any $\aaa, \ii \in B_l, \ii \in B_m$ and $\jj \in B_{m-l}$ with $l<m$, the following equality is valid:
\begin{equation}
\label{t:ms}
M(q,t)_{\ii,\jj}^{\aaa,\bb} = S(q,t)_{\ii,\jj}^{\aaa,\bb}.
\end{equation}
\end{theorem}

\begin{proof}
In view of (\ref{a:m1}) and (\ref{a:S1}), we shall exclusively consider the non-trivial situation $\aaa+\bb=\jj$.
We prove~\eqref{t:ms} by comparing the recursion relations with respect to $l$ on both hand sides.
Let us illustrate the derivation for $S(q,t)_{\ii,\jj}^{\aaa,\bb}$ using the example
$\ii=(000000110100), \aaa=(100010000100) \in B_{l=3}$,
$\jj=(100110001101) \in B_{m=6}$ with $\bb=\jj-\aaa$ and $L=12$.
According to (\ref{a:abij}), the ball diagram relevant to 
$M(q,t)_{\ii,\jj}^{\aaa,\bb}$ looks as
\begin{align}\label{a:bd}
\begin{tikzpicture}[scale=0.7]
\foreach \x in {4,9,12}
  \draw (\x, 0) circle (.25);
  \foreach \x in {1,5,10}
  \draw[fill=red!70] (\x,0) circle (0.25);
\foreach \x in {2,3,6,7,8,11}
  \fill[black] (\x, 0) circle (.1);
  \foreach \x in {7,8,10}
  \draw (\x, -1) circle (.25);
\foreach \x in {1,2,3,4,5,6,9,11,12}
  \fill[black] (\x, -1) circle (.1);
  \draw(1,-2) node{$\aaa^-$};
\draw(4,-2) node{$\kk$};
\draw(5,-2) node{$\aaa^-$};
\draw(7,-2) node{$\aaa^+$}; 
\draw(8,-2) node{$\aaa^+$};
\draw(9,-2) node{$\kk$};
\draw(10,-2) node{$1$};
\draw(12,-2) node{$\kk$};
\end{tikzpicture}
\end{align}
where the upper and lower rows stand for $\jj$ and $\ii$, respectively, and 
$\aaa$ is depicted as the red shaded balls.\footnote{This red shading has nothing to do with the color $\mathcal{C}$ of balls 
assigned in the MLQs explained in Section \ref{subsec:mq}.}
We have also presented $S^{a_1 b_1}_{i_1 j_1}, \ldots, S^{a_L b_L}_{i_L j_L}$ that are not 1 at the bottom,
except for the third column from the right, for later convenience.
Thus (\ref{t:Sop}) reads 
\begin{align*}
S(q,t)_{\ii,\jj}^{\aaa,\bb} = (1-qt^3)\Theta_0, \qquad 
\Theta_0 = \tr(q^\hh \aaa^- \kk\, \aaa^- \aaa^+ \aaa^+ \kk\, 1 \kk).
\end{align*}
What we do is to move the leftmost $\aaa^+$  here to the left 
cyclically to go around the trace once by using the 
following form of the relations (\ref{a:aak}) for $\mathcal{A}$:
\begin{subequations}
\label{a:qoa}
\begin{align}
\aaa^- \aaa^+ & = t \aaa^+ \aaa^- + (1 - t), 
\label{a:aat}\\
\kk \,\aaa^+ & = t \aaa^+ \kk, \\
q^\hh \aaa^+ & = q\,\aaa^+ q^\hh.
\end{align}
\end{subequations}
By applying (\ref{a:aat}), the trace $\Theta_0$ is decomposed as
\begin{equation}
\Theta_0=\Theta'_0+\Theta_1,
\qquad 
\Theta'_0 = (1-t) \tr(q^\hh \aaa^- \kk\, \aaa^+ \kk\, 1 \kk) ,
\qquad
\Theta_1 = t \tr(q^\hh \aaa^- \kk\, \aaa^+ \aaa^- \aaa^+ \kk\, 1 \kk).
\end{equation}
The term $\Theta'_0$ corresponds to the following diagram where the connected balls should be understood absent: 
\begin{align}\label{a:bd0}
\begin{tikzpicture}[scale=0.7]
\draw(5,1) node {$1\!-\!t$};
\foreach \x in {4,9,12}
  \draw (\x, 0) circle (.25);
  \foreach \x in {1,5,10}
  \draw[fill=red!70](\x,0) circle (0.25);
\foreach \x in {2,3,6,7,8,11}
  \fill[black] (\x, 0) circle (.1);
  \foreach \x in {7,8,10}
  \draw (\x, -1) circle (.25);
\foreach \x in {1,2,3,4,5,6,9,11,12}
  \fill[black] (\x, -1) circle (.1);
  \draw(1,-2) node{$\aaa^-$};
\draw(4,-2) node{$\kk$};
\draw(8,-2) node{$\aaa^+$};
\draw(9,-2) node{$\kk$};
\draw(10,-2) node{$1$};
\draw(12,-2) node{$\kk$};
\draw[rounded corners,thick,->,color=black] (7,-0.76) -- (7,-0.5) -- (5,-0.5)--(5,-0.24);
\end{tikzpicture}
\end{align}
The trace
$\Theta_1$ is decomposed as
\begin{equation}
\begin{split}
\Theta_1 &= t^2 \tr(q^\hh \aaa^- \aaa^+ \kk\, \aaa^- \aaa^+ \kk\, 1 \kk) = \Theta'_1+\Theta_2,
\\
\Theta'_1&= t^2(1-t) \tr(q^\hh \kk\,  \aaa^- \aaa^+ \kk\, 1 \kk), \quad 
\Theta_2 = t^3\tr(q^\hh \aaa^+ \aaa^-\kk\, \aaa^- \aaa^+ \kk\, 1\kk).
\end{split}
\end{equation}
The term $\Theta'_1$ corresponds to the following diagram where the connected balls should be understood absent: 
\begin{align}\label{a:bd1}
\begin{tikzpicture}[scale=0.7]
\draw(1,1) node {$1\!-\!t$};
\draw(4,1) node {$t$};
\draw(5,1) node {$t$};
\foreach \x in {4,9,12}
  \draw (\x, 0) circle (.25);
  \foreach \x in {1,5,10}
  \draw[fill=red!70](\x,0) circle (0.25);
\foreach \x in {2,3,6,7,8,11}
  \fill[black] (\x, 0) circle (.1);
  \foreach \x in {7,8,10}
  \draw (\x, -1) circle (.25);
\foreach \x in {1,2,3,4,5,6,9,11,12}
  \fill[black] (\x, -1) circle (.1);
\draw(4,-2) node{$\kk$};
\draw(5,-2) node{$\aaa^-$};
\draw(8,-2) node{$\aaa^+$};
\draw(9,-2) node{$\kk$};
\draw(10,-2) node{$1$};
\draw(12,-2) node{$\kk$};
\draw[rounded corners,thick,->,color=black] (7,-0.76) -- (7,-0.5) -- (1,-0.5)--(1,-0.24);
\end{tikzpicture}
\end{align}
The trace  $\Theta_2$ is decomposed by rewriting the operator $1$ as $(1-t)+t$ as 
\begin{equation}
\begin{split}
\Theta_2 &= qt^4 \tr(q^\hh \aaa^- \kk\, \aaa^- \aaa^+ \kk \, 1 \ap \kk) = \Theta'_2+\Theta_3,
\\
\Theta'_2 &= qt^4(1-t) \tr(q^\hh\aaa^-  \kk\,   \aaa^- \aaa^+ \kk\, \aaa^+ \kk),
\quad
\Theta_3 = qt^6 \tr(q^\hh \aaa^- \kk\, \aaa^- \aaa^+ \aaa^+ \kk\,\kk) = qt^6 \Theta_0.
\end{split}
\end{equation}
The term $\Theta'_2$ corresponds to the following diagram where the connected balls should be understood absent: 
\begin{align}\label{a:bd2}
\begin{tikzpicture}[scale=0.7]
\draw(0,1) node {$q$};
\draw(1,1) node {$t$};
\draw(4,1) node {$t$};
\draw(5,1) node {$t$};
\draw(10,1) node {$1\!-\!t$};
\draw(12,1) node {$t$};
\foreach \x in {4,9,12}
  \draw (\x, 0) circle (.25);
  \foreach \x in {1,5,10}
  \draw[fill=red!70](\x,0) circle (0.25);
\foreach \x in {2,3,6,7,8,11}
  \fill[black] (\x, 0) circle (.1);
  \foreach \x in {7,8,10}
  \draw (\x, -1) circle (.25);
\foreach \x in {1,2,3,4,5,6,9,11,12}
  \fill[black] (\x, -1) circle (.1);
 \draw(1,-2) node{$\aaa^-$};
\draw(4,-2) node{$\kk$};
\draw(5,-2) node{$\aaa^-$};
\draw(8,-2) node{$\aaa^+$};
\draw(9,-2) node{$\kk$};
\draw(10,-2) node{$\aaa^+$};
\draw(12,-2) node{$\kk$};
\draw[rounded corners,thick,-] (7,-0.76) -- (7,-0.5) -- (0.5,-0.5);
\draw[dotted,thick] (0.5,-0.5) -- (0,-0.5);
\draw[dotted,thick] (13,-0.5) -- (12.5,-0.5);
\draw[rounded corners,thick,->] (12.5,-0.5) -- (10,.-0.5) -- (10,-0.24);
\end{tikzpicture}
\end{align}
Thus we have the relation
\begin{equation*}
(1-qt^6) \tr(q^\hh \aaa^- \kk\, \aaa^- \aaa^+ \aaa^+ \kk\, \kk)
= (1-qt^6) \Theta_0 
=\Theta'_0 + \Theta'_1+\Theta'_2.
\end{equation*}
The traces in $\Theta'_\alpha$ involve operators obtained by eliminating 
$\ap \am$, with (\ref{a:aat}) replaced by $1-t$ if the target of the pairing 
is free from a ball directly below it. If the target has a ball underneath, 
the corresponding operator $1$ is replaced by $\ap$.
The coefficient $t$ (resp.\ $q$) in the RHS of (\ref{a:qoa}) counts the number of skipped balls
(resp.\ wrapping). 
In general an analogous manipulation yields
\begin{align}\label{a:LP}
(1-qt^m)\tr \bigl(q^\hh S_{i_1j_1}^{a_1b_1} S_{i_2j_2}^{a_2b_2} \cdots S_{i_Lj_L}^{a_Lb_L} \bigr)
= \sum_{\text{leftmost pairings}}\!\!
(1-t) q^{\#\text{wrapped}}t^{\#\text{skipped}}
\tr(\text{unpaired parts}),
\end{align}
where the sum extends over the leftmost pairings, which means those between the leftmost $\aaa^+$ in the downstairs and 
the red shaded balls, i.e., $\aaa \in B_l$, in the upstairs. Thus there are $\abs{\aaa} = l$ summands.
Let $\aaa', \ii'  \in B_{l-1},  \jj'  \in B_{m-1}$ and $\jj' \in B_{(m-1)-(l-1)}=B_{m-l}$ 
be the arrays corresponding to the unpaired parts.
Then from the definition (\ref{t:Sop}),   the  recursion relation (\ref{a:LP}) 
is translated as
\begin{align}
S(q,t)^{\aaa, \bb}_{\ii,\jj} 
= \frac{1-t}{1-qt^m}\sum_{\text{leftmost pairings}}
q^{\#\text{wrapped}}t^{\#\text{skipped}}
S(q,t)^{\aaa', \bb'}_{\ii',\jj'}.
\end{align}
On the other hand,  from the queuing algorithm explained in the previous section,
it is clear that the same recursion relation holds also for $M(q,t)^{\aaa, \bb}_{\ii,\jj}$.
In particular, the factor $(1-qt^m)^{-1}$ emerges from the fact that 
the denominators $1-qt^{\#\text{free}}$ in (\ref{a:wt2}) lead to 
$\prod_{r=1}^l(1-qt^{m+1-r})$ for $M(q,t)^{\aaa, \bb}_{\ii,\jj}$ 
whereas $\prod_{r=1}^{l-1}(1-qt^{m-r})$ for $M(q,t)^{\aaa', \bb'}_{\ii',\jj'}$.
It follows, by induction on $l$, the proof of (\ref{t:ms}) reduces formally to $l=0$.
When $l=0$, one has $M(q,t)^{\aaa, \bb}_{\ii,\jj}=\delta^\bb_\jj$ from (\ref{a:m1}), 
and $S(q,t)^{\aaa, \bb}_{\ii,\jj} = (1-qt^m)\tr(q^\hh \kk^m)\delta^\bb_\jj=\delta^\bb_\jj$ from
(\ref{t:s5V}) and (\ref{t:Sop}).
This completes the proof.
\end{proof}

Recall from~\eqref{a:m1} that $M(q,t)^{\aaa, \bb}_{\ii,\jj}$ consists of 
numerous summands, 
which have to be calculated resorting to pairing diagrams by counting the skipped/free balls and wrapping.
Theorem~\ref{t:th:ms} identifies it with a single trace that is free from such pairing details.
Put in the other way, $M(q,t)^{\aaa, \bb}_{\ii,\jj}$ was giving a combinatorial description for an expansion 
of the trace $S(q,t)^{\aaa, \bb}_{\ii,\jj}$. 

In view of Theorem~\ref{t:th:ms}, the interpretation of the multi-indices $\aaa, \bb, {\bf i}, {\bf j}$ in (\ref{a:abij}) 
and the layer structure in (\ref{a:fig10}),  the five vertex weights in (\ref{t:s5V}) describes 
 a \emph{local} event in the queuing process with respect to the ``time" $\alpha \in \Z_L$,
 where the number of the custmors in the queue is increased or decreased by $\ap_\alpha$ or $\am_\alpha$. 

\subsection{Operator $\mathbb{S}(q,t)$}\label{subsec:S}
A key ingredient in Proposition~\ref{a:pr:P} is the operator $\mathbb{M}(q,t)$ whose 
building block was $\widecheck{M}(q,t)$ in (\ref{eq:MVV}).
From Theorem~\ref{t:th:ms}, it is natural to reformulate them in terms of 
$S(q,t)^{\aaa, \bb}_{\ii, \jj} $.
This is what we do in this subsection as a preparation for the introduction of the operator $X_\alpha(z)$ 
in the next subsection.

For $l < m$, define a linear operator $\widecheck{S}(q,t)$ in parallel with $\widecheck{M}(q,t)$ in (\ref{eq:MVV}):
\begin{subequations}
\label{eq:SVV}
\begin{align}
\widecheck{S}(q,t)\colon V_l \otimes V_m & \rightarrow V_{m-l} \otimes V_l
\label{a:Svv0} \\
v_\ii \otimes v_\jj \;&\mapsto
\sum_{\aaa \otimes \bb \in B_l\otimes B_{m-l}}
S(q,t)^{\aaa, \bb}_{\ii, \jj} \,v_{\bb} \otimes v_\aaa = \sum_{\aaa \in B_l, \aaa\le \jj} S(q,t)^{\aaa, \jj-\aaa}_{\ii, \jj} \,v_{\jj-\aaa} \otimes v_\aaa
\qquad 
(\ii \otimes \jj \in B_l \otimes B_m),
\label{a:Svv}
\end{align}
\end{subequations}
where the last equality in~\eqref{a:Svv} is due to (\ref{a:S1}).
By Theorem~\ref{t:th:ms} we know
\begin{align}
\widecheck{M}(q,t) = \widecheck{S}(q,t).
\end{align}
Substituting this into~\eqref{a:mod} we have
 \begin{equation}\label{a:mod2}
 \begin{split}
 \mathbb{M}(q,t) = \mathbb{S}(q,t):=\qquad \qquad \qquad \qquad \qquad \qquad 
  \widecheck{S}(q)_{n,n-1}& \\
 \times  \widecheck{S}(q^2)_{n-1,n-2} \widecheck{S}(q)_{n,n-1} & \\
 \cdots \qquad \\
 \times \widecheck{S}(q^{n-1})_{2,1}  \widecheck{S}(q^{n-2})_{3,2} 
 \cdots \widecheck{S}(q)_{n,n-1}.
 \end{split}
 \end{equation} 
From these definitions,  Proposition~\ref{a:pr:P} can be restated as 
\begin{corollary}\label{t:co:P}
\begin{align}\label{a:pm}
 \ket{P_{\text{MLQ}}(\mm)_q} = 
 \Pi\Bigl(
 \mathbb{S}(q,t)\sum_{\bb_n \otimes \cdots \otimes \bb_1 \in B(\mm)}
 v_{\bb_n} \otimes \cdots \otimes v_{\bb_1} 
 \Bigr).
\end{align}
\end{corollary}

\subsection{Matrix product operators $X_\alpha(z)$}

In the remainder of this section we concentrate on the $q=1$ case, which is relevant to the actual stationary states on the $n$-ASEP.
We do not exhibit $z$ in the vertex diagram~\eqref{t:vm} and $q$ in the diagrams like~\eqref{eq:f7} and~\eqref{eq:f8}
assuming that they are all set to $1$.
We depict \eqref{t:Sop} as
\begin{equation}\label{a:fig9}
S(1,t)^{\aaa,\bb}_{\ii,\jj} = (1 - t^{m-\ell}) {\color{blue}\tr_F} \left(
\begin{tikzpicture}[scale=.65,baseline=-5]
\draw[<-] (-30:.9) node[anchor=north west] {$a_1$} -- (150:.9) node[anchor=south east] {$i_1$};
\draw[->] (0,-1) node[anchor=north]{$j_1$} -- (0,1) node[anchor=south] {$b_1$};
\draw[<-] ++(2,0) + (-30:.9) node[anchor=north west] {$a_2$} -- +(150:.9) node[anchor=south east] {$i_2$};
\draw[->] (2,-1) node[anchor=north]{$j_2$} -- (2,1) node[anchor=south] {$b_2$};
\draw[<-] ++(6,0) + (-30:.9) node[anchor=north west] {$a_L$} -- +(150:.9) node[anchor=south east] {$i_L$};
\draw[->] (6,-1) node[anchor=north]{$j_L$} -- (6,1) node[anchor=south] {$b_L$};
\draw[-,blue] (7,0) -- (5,0);
\draw[dashed,blue] (5,0) -- (3,0);
\draw[->,blue] (3,0) -- (-1,0);
\end{tikzpicture}
\right).
\end{equation}
Here each vertex signifies $S^{a_r b_r}_{i_r j_r}$ taking values in the $t$-oscillator algebra as in~\eqref{t:s5V}.
Blue arrows are added to signify the Fock space $F$ (\ref{a:F}) on which it acts and the trace is taken.
The operator $\widecheck{S}(1,t)$ from~\eqref{eq:SVV} will similarly be depicted by suppressing the indices.

The operator $\mathbb{S}(1,t)$ shares the same corner transfer matrix (CTM) diagram representation 
as $\mathbb{M}(1,t)$ in (\ref{eq:f8}), with each vertex (\ref{t:vm}) of thick arrows 
now possessing the structure in the ``third dimension" as illustrated in (\ref{a:fig9}).
As the result, the stationary state in Corollary~\ref{t:co:P} is given as  
\begin{align}
 \ket{P_{\text{MLQ}}(\mm)_{q=1}} 
= \sum_{{\boldsymbol \sigma} \in \Sigma(\mm)}
\stprob({\boldsymbol \sigma}) \ket{ {\boldsymbol \sigma} },
\end{align}
where the (unnormalized) stationary probability $\stprob({\boldsymbol \sigma})$, for example for $n=3$, 
has the diagram representation:

\begin{equation}
\label{a:fig10}
\prob({\boldsymbol\sigma}) = {\color{blue}\tr_{F^{\otimes 3}}} \left(\,
\begin{tikzpicture}[scale=.8,baseline=1.2cm]
\draw[->,rounded corners] (0,0.3) -- ++(0,2.7) -- ++(-30:.7);
\draw[->,rounded corners] ++(150:.9) + (0,0.3) -- ++(0,2) -- ++(-30:1.6);
\draw[->,rounded corners] ++(150:1.8) + (0,0.3) -- ++(0,1) -- ++(-30:2.5);
\draw[->,rounded corners] ++(5,0.3) -- ++(0,2.7) -- ++(-30:.7);
\draw[->,rounded corners] ++(5,0) ++(150:.9) + (0,0.2) -- ++(0,2) -- ++(-30:1.6);
\draw[->,rounded corners] ++(5,0) ++(150:1.8) + (0,0.2) -- ++(0,1) -- ++(-30:2.5);
\draw[-,blue] (5.6,1) -- (3.5,1);
\draw[dashed,blue] (3.5,1) -- (1.5,1);
\draw[->,blue] (1.5,1) -- (-.6,1);
\draw[-,blue] (5.6,2) -- (3.5,2);
\draw[dashed,blue] (3.5,2) -- (1.5,2);
\draw[->,blue] (1.5,2) -- (-.6,2);
\draw[-,blue] ++(0,2) ++ (150:.9) + (5.6,-1) -- +(3.5,-1);
\draw[dashed,blue] ++(0,2) ++ (150:.9) + (3.5,-1) -- +(1.5,-1);
\draw[->,blue] ++(0,2) ++ (150:.9) + (1.5,-1) -- +(-.6,-1);
\end{tikzpicture}
\;\right),
\end{equation}
where there are $L$ layers of the CTM.

The dependence on ${\boldsymbol \sigma}=(\sigma_1, \ldots, \sigma_L)$ 
is reflected in the boundary condition.
In fact, comparison of (\ref{a:fig10}) and (\ref{a:mho}) amounts to identifying the 
$i$th layer with $X_{\sigma_i}$.
Moreover, inspecting the projection $\Pi$ in (\ref{a:Pi}) leads to the constraint that 
$X_\alpha\, (\alpha = 0,\ldots, n)$ is a partition function of the five vertex model 
in the NW quadrant 
with a free boundary condition at the bottom and a fixed one on the right, 
where the edge states $s_1,\ldots, s_n$ on the outgoing arrows (numbered from the top) 
are specified as $s_i=\delta_{i,\alpha}$. 
Such a CTM diagram representation of the stationary probabilities was first obtained for TASEP 
(the $t=0$ case)  in~\cite{KMO1,KMO2} using the five vertex model in Remark~\ref{a:re:st}.
See the recent work~\cite{IMO} for a new application.

Recall from Section~\ref{a:subsec:mp} that one also needs the spectral parameter dependent version 
$X_\alpha(z)$ of $X_\alpha$ in order to establish the stationary condition 
by resorting to the ZF algebra (\ref{a:zf})  independently from the MLQ construction.

\begin{definition}[The operator $X_{\alpha}(z)$]
We introduce the operator $X_\alpha(z)\, (\alpha = 0,\ldots, n)$ 
as the CTM with each vertex having a $t$-oscillator valued weight specified in~\eqref{t:s5V} and depicted as
\begin{equation}\label{a:ef12}
X_{\alpha}(z) = 
\begin{tikzpicture}[scale=0.8,>=latex,baseline=-2.5cm,rotate=270]
\foreach \x in {1,2,3} {
	\draw[->, rounded corners] (3.5,-\x)  -- (\x,-\x) -- (\x,-0.2) node[anchor=west] {$\delta_{\alpha\x}$};
	\draw (3.8, -\x) node {$\vdots$};
	\draw[-] (4.5,-\x) -- (5.8,-\x);
	\fill[black] (5.45,-\x) circle (.1) node[anchor=west, scale=.8] {$z^{\hh}$};
}
\draw[->, rounded corners] (5.8,-5)  -- (5,-5) -- (5,-0.2) node[anchor=west] {$\delta_{\alpha n}$};
%\draw[->, rounded corners] (6.8,-5)  -- (6,-5) -- (6,-0.2) node[anchor=west] {$\delta_{\alpha n}$};
\fill[black] (5.45,-5) circle (.1) node[anchor=west, scale=.8] {$z^{\hh}$};
\draw (4,-4) node {\rotatebox{90}{$\ddots$}};
\draw (5.5,-3.8) node {$\cdots$};
\draw (3.8,-0) node {$\vdots$};
\end{tikzpicture}
\end{equation}
where the $z^{\hh}$ are denoted by $\bullet$.
The diagram indicates that all edge variables are summed over $\{0,1\}$. 
The variables on the right outgoing edges are fixed as shown, 
while there is no boundary condition along the bottom.
This depiction should be understood as having an additional arrow 
at each vertex, perpendicular to the layer, corresponding to the 
$t$-oscillator acting on its own Fock space.
\end{definition}

\begin{example}
The operators $X_{\alpha}(z)$ can be written as
\begin{equation}\label{a:ef11}
\begin{array}{c@{\hspace{50pt}}c}
\mathop{\quad\scalebox{1.5}{$\sum\;\;$}}\limits_{i_1,i_2=0}^1  z^{i_1 + i_2}
\begin{tikzpicture}[scale=0.65,>=latex,baseline=-38,rotate=270]
\foreach \x in {1,2}
	\draw[->, rounded corners] (2.7,-\x) node[anchor=north] {$i_{\x}$} 
	-- (\x,-\x) -- (\x,-0.2) node[anchor=west] {$\delta_{\alpha\x}$};
\end{tikzpicture}
&
\mathop{\quad\scalebox{1.5}{$\sum\;\;$}}\limits_{i_1,i_2,i_3=0}^1 z^{i_1 + i_2 + i_3}
\begin{tikzpicture}[scale=0.65,>=latex,baseline=-50,rotate=-90]
\foreach \x in {1,2,3}
	\draw[->, rounded corners] (3.7,-\x) node[anchor=north] {$i_{\x}$}
	-- (\x,-\x) -- (\x,-0.2) node[anchor=west] {$\delta_{\alpha\x}$};
\end{tikzpicture}
\\
n = 2 & n = 3
\end{array}
\end{equation}
\end{example}

\noindent
Obviously this $X_\alpha(z)$ reduces to the $X_\alpha$ considered previously when $z=1$.
It is an element of $\mathrm{End}(F^{\otimes \frac{n(n-1)}{2}})$.
We number the vertices in (\ref{a:ef12})  
from top to bottom in the rightmost column, and then similarly in the second rightmost column, and so on, as 
$1,\ldots, \frac{n(n-1)}{2}$ in this order.
The $t$-oscillator generators acting on the $i$th copy of the Fock space $F$ in this numbering
will be distinguished by the subscript $i$ when $n \ge 3$.
Thus, the generators with different subscripts commute.
To summarize, we have proved the following.

\begin{theorem}[Five vertex CTM interpretation of the MLQ construction]\label{a:th:pxx}
The MLQ construction  in Section~\ref{sec:mlq} yields the matrix product formula~\eqref{a:mho} of the unnormalized stationary probability, where the operators $X_0,\ldots, X_n$ are given as $X_\alpha = X_\alpha(z=1)$ for $X_\alpha(z)$ defined by~\eqref{a:ef12} and the trace is taken over $F^{\otimes \frac{n(n-1)}{2}}$.
\end{theorem}

In the next section, we will establish the same matrix product formula by the ZF algebra without relying on the MLQ construction.

\begin{example}\label{a:ex:X}
For $n=2$, $X_\alpha(z)$ is given by
\begin{align*}
\raisebox{2pt}{$X_0(z)$} & \begin{array}{@{\;=}c@{\;+\;}c} \begin{tikzpicture}[scale=.65,>=latex,rounded corners,baseline=-1cm,rotate=-90]
\draw[->] (1.7,0) node[anchor=north] {$0$} -- (0,0) -- (0,0.7) node[anchor=west] {$0$};
\draw[->] (1.7,-1) node[anchor=north] {$0$} -- (1,-1) -- (1,0.7) node[anchor=west] {$0$};
\end{tikzpicture}
& \begin{tikzpicture}[scale=.65,>=latex,rounded corners,baseline=-1cm,rotate=-90]
\draw[->] (1.7,0) node[anchor=north] {$0$} -- (0,0) -- (0,0.7) node[anchor=west] {$0$};
\draw[->] (1.7,-1) node[anchor=north] {$1$} -- (1,-1) -- (1,0.7) node[anchor=west] {$0$};
\end{tikzpicture}
z
\\
1 & z \aaa^+,
\end{array}
&
\raisebox{2pt}{$X_1(z)$} & \begin{array}{@{\;=}c} \begin{tikzpicture}[scale=.65,>=latex,rounded corners,baseline=-1cm,rotate=-90]
\draw[->] (1.7,0) node[anchor=north] {$1$} -- (0,0) -- (0,0.7) node[anchor=west] {$1$};
\draw[->] (1.7,-1) node[anchor=north] {$0$} -- (1,-1) -- (1,0.7) node[anchor=west] {$0$};
\end{tikzpicture}
z
\\
z \kk,
\end{array}
&
\raisebox{2pt}{$X_2(z)$} & \begin{array}{@{\;=}c@{\;+\;}c} \begin{tikzpicture}[scale=.65,>=latex,rounded corners,baseline=-1cm,rotate=-90]
\draw[->] (1.7,0) node[anchor=north] {$1$} -- (0,0) -- (0,0.7) node[anchor=west] {$0$};
\draw[->] (1.7,-1) node[anchor=north] {$0$} -- (1,-1) -- (1,0.7) node[anchor=west] {$1$};
\end{tikzpicture}
z
& \begin{tikzpicture}[scale=.65,>=latex,rounded corners,baseline=-1cm,rotate=-90]
\draw[->] (1.7,0) node[anchor=north] {$1$} -- (0,0) -- (0,0.7) node[anchor=west] {$0$};
\draw[->] (1.7,-1) node[anchor=north] {$1$} -- (1,-1) -- (1,0.7) node[anchor=west] {$1$};
\end{tikzpicture}
z^2
\\
z \aaa^- & z^2.
\end{array}
\end{align*}These formulas agree with~\cite[Eq.(46)]{CDW} 
under the formal transformation $(\ap, \am, \kk, z) \rightarrow (a, a^\dagger,k, x)$.
For $n=3$, $X_\alpha(z)$ is given by
\begin{align*}
\raisebox{9pt}{$X_0(z)$} & \begin{array}{@{\;=}c@{\;+\;}c@{\;+\;}c@{\;+\;}c@{\;+\;}c} \begin{tikzpicture}[scale=.6,>=latex,rounded corners,baseline=-1cm,rotate=-90]
\draw[->] (2.7,0) node[anchor=north] {$0$} -- (0,0) -- (0,0.7) node[anchor=west] {$0$};
\draw[->] (2.7,-1) node[anchor=north] {$0$} -- (1,-1) -- (1,0.7) node[anchor=west] {$0$};
\draw[->] (2.7,-2) node[anchor=north] {$0$} -- (2,-2) -- (2,0.7) node[anchor=west] {$0$};
\draw (1.5,0) node[anchor=west,scale=.7] {$0$};
\draw (1.5,-1) node[anchor=west,scale=.7] {$0$};
\draw (2,-.5) node[anchor=north,scale=.7] {$0$};
\end{tikzpicture}
&
\begin{tikzpicture}[scale=.6,>=latex,rounded corners,baseline=-1cm,rotate=-90]
\draw[->] (2.7,0) node[anchor=north] {$0$} -- (0,0) -- (0,0.7) node[anchor=west] {$0$};
\draw[->] (2.7,-1) node[anchor=north] {$1$} -- (1,-1) -- (1,0.7) node[anchor=west] {$0$};
\draw[->] (2.7,-2) node[anchor=north] {$0$} -- (2,-2) -- (2,0.7) node[anchor=west] {$0$};
\draw (1.5,0) node[anchor=west,scale=.7] {$0$};
\draw (1.5,-1) node[anchor=west,scale=.7] {$1$};
\draw (2,-.5) node[anchor=north,scale=.7] {$0$};
\end{tikzpicture}
z
&
\begin{tikzpicture}[scale=.6,>=latex,rounded corners,baseline=-1cm,rotate=-90]
\draw[->] (2.7,0) node[anchor=north] {$0$} -- (0,0) -- (0,0.7) node[anchor=west] {$0$};
\draw[->] (2.7,-1) node[anchor=north] {$1$} -- (1,-1) -- (1,0.7) node[anchor=west] {$0$};
\draw[->] (2.7,-2) node[anchor=north] {$0$} -- (2,-2) -- (2,0.7) node[anchor=west] {$0$};
\draw (1.5,0) node[anchor=west,scale=.7] {$0$};
\draw (1.5,-1) node[anchor=west,scale=.7] {$0$};
\draw (2,-.5) node[anchor=north,scale=.7] {$1$};
\end{tikzpicture}
z
&
\begin{tikzpicture}[scale=.6,>=latex,rounded corners,baseline=-1cm,rotate=-90]
\draw[->] (2.7,0) node[anchor=north] {$0$} -- (0,0) -- (0,0.7) node[anchor=west] {$0$};
\draw[->] (2.7,-1) node[anchor=north] {$0$} -- (1,-1) -- (1,0.7) node[anchor=west] {$0$};
\draw[->] (2.7,-2) node[anchor=north] {$1$} -- (2,-2) -- (2,0.7) node[anchor=west] {$0$};
\draw (1.5,0) node[anchor=west,scale=.7] {$0$};
\draw (1.5,-1) node[anchor=west,scale=.7] {$0$};
\draw (2,-.5) node[anchor=north,scale=.7] {$0$};
\end{tikzpicture}
z
&
\begin{tikzpicture}[scale=.6,>=latex,rounded corners,baseline=-1cm,rotate=-90]
\draw[->] (2.7,0) node[anchor=north] {$0$} -- (0,0) -- (0,0.7) node[anchor=west] {$0$};
\draw[->] (2.7,-1) node[anchor=north] {$1$} -- (1,-1) -- (1,0.7) node[anchor=west] {$0$};
\draw[->] (2.7,-2) node[anchor=north] {$1$} -- (2,-2) -- (2,0.7) node[anchor=west] {$0$};
\draw (1.5,0) node[anchor=west,scale=.7] {$0$};
\draw (1.5,-1) node[anchor=west,scale=.7] {$0$};
\draw (2,-.5) node[anchor=north,scale=.7] {$1$};
\end{tikzpicture}
z^2
\\
1 & z \aaa_1^+ \kk_3 & z \aaa_2^+ \aaa_3^- & z \aaa_3^+ & z^2 \aaa_2^+,
\end{array}
\allowdisplaybreaks \\
\raisebox{9pt}{$X_1(z)$} & \begin{array}{@{\;=}c@{\;+\;}c}  \begin{tikzpicture}[scale=.6,>=latex,rounded corners,baseline=-1cm,rotate=-90]
\draw[->] (2.7,0) node[anchor=north] {$1$} -- (0,0) -- (0,0.7) node[anchor=west] {$1$};
\draw[->] (2.7,-1) node[anchor=north] {$0$} -- (1,-1) -- (1,0.7) node[anchor=west] {$0$};
\draw[->] (2.7,-2) node[anchor=north] {$0$} -- (2,-2) -- (2,0.7) node[anchor=west] {$0$};
\draw (1.5,0) node[anchor=west,scale=.7] {$1$};
\draw (1.5,-1) node[anchor=west,scale=.7] {$0$};
\draw (2,-.5) node[anchor=north,scale=.7] {$0$};
\end{tikzpicture}
z
&
\begin{tikzpicture}[scale=.6,>=latex,rounded corners,baseline=-1cm,rotate=-90]
\draw[->] (2.7,0) node[anchor=north] {$1$} -- (0,0) -- (0,0.7) node[anchor=west] {$1$};
\draw[->] (2.7,-1) node[anchor=north] {$0$} -- (1,-1) -- (1,0.7) node[anchor=west] {$0$};
\draw[->] (2.7,-2) node[anchor=north] {$1$} -- (2,-2) -- (2,0.7) node[anchor=west] {$0$};
\draw (1.5,0) node[anchor=west,scale=.7] {$1$};
\draw (1.5,-1) node[anchor=west,scale=.7] {$0$};
\draw (2,-.5) node[anchor=north,scale=.7] {$0$};
\end{tikzpicture}
z^2
\\
z \kk_1 \kk_2 & z^2 \kk_1 \kk_2 \aaa_3^+,
\end{array}
\allowdisplaybreaks \\
\raisebox{9pt}{$X_2(z)$} & \begin{array}{@{\;=}c@{\;+\;}c@{\;+\;}c} \begin{tikzpicture}[scale=.6,>=latex,rounded corners,baseline=-1cm,rotate=-90]
\draw[->] (2.7,0) node[anchor=north] {$1$} -- (0,0) -- (0,0.7) node[anchor=west] {$0$};
\draw[->] (2.7,-1) node[anchor=north] {$0$} -- (1,-1) -- (1,0.7) node[anchor=west] {$1$};
\draw[->] (2.7,-2) node[anchor=north] {$0$} -- (2,-2) -- (2,0.7) node[anchor=west] {$0$};
\draw (1.5,0) node[anchor=west,scale=.7] {$1$};
\draw (1.5,-1) node[anchor=west,scale=.7] {$0$};
\draw (2,-.5) node[anchor=north,scale=.7] {$0$};
\end{tikzpicture}
z
&
\begin{tikzpicture}[scale=.6,>=latex,rounded corners,baseline=-1cm,rotate=-90]
\draw[->] (2.7,0) node[anchor=north] {$1$} -- (0,0) -- (0,0.7) node[anchor=west] {$0$};
\draw[->] (2.7,-1) node[anchor=north] {$0$} -- (1,-1) -- (1,0.7) node[anchor=west] {$1$};
\draw[->] (2.7,-2) node[anchor=north] {$1$} -- (2,-2) -- (2,0.7) node[anchor=west] {$0$};
\draw (1.5,0) node[anchor=west,scale=.7] {$1$};
\draw (1.5,-1) node[anchor=west,scale=.7] {$0$};
\draw (2,-.5) node[anchor=north,scale=.7] {$0$};
\end{tikzpicture}
z^2
&
\begin{tikzpicture}[scale=.6,>=latex,rounded corners,baseline=-1cm,rotate=-90]
\draw[->] (2.7,0) node[anchor=north] {$1$} -- (0,0) -- (0,0.7) node[anchor=west] {$0$};
\draw[->] (2.7,-1) node[anchor=north] {$1$} -- (1,-1) -- (1,0.7) node[anchor=west] {$1$};
\draw[->] (2.7,-2) node[anchor=north] {$0$} -- (2,-2) -- (2,0.7) node[anchor=west] {$0$};
\draw (1.5,0) node[anchor=west,scale=.7] {$1$};
\draw (1.5,-1) node[anchor=west,scale=.7] {$1$};
\draw (2,-.5) node[anchor=north,scale=.7] {$0$};
\end{tikzpicture}
z^2
\\
z \aaa_1^- \kk_2 & z^2 \aaa_1^- \kk_2 \aaa_3^+ & z^2 \kk_2 \kk_3,
\end{array}
\allowdisplaybreaks \\
\raisebox{9pt}{$X_3(z)$} & \begin{array}{@{\;=}c@{\;+\;}c@{\;+\;}c@{\;+\;}c@{\;+\;}c} \begin{tikzpicture}[scale=.6,>=latex,rounded corners,baseline=-1cm,rotate=-90]
\draw[->] (2.7,0) node[anchor=north] {$1$} -- (0,0) -- (0,0.7) node[anchor=west] {$0$};
\draw[->] (2.7,-1) node[anchor=north] {$0$} -- (1,-1) -- (1,0.7) node[anchor=west] {$0$};
\draw[->] (2.7,-2) node[anchor=north] {$0$} -- (2,-2) -- (2,0.7) node[anchor=west] {$1$};
\draw (1.5,0) node[anchor=west,scale=.7] {$0$};
\draw (1.5,-1) node[anchor=west,scale=.7] {$0$};
\draw (2,-.5) node[anchor=north,scale=.7] {$0$};
\end{tikzpicture}
z
&
\begin{tikzpicture}[scale=.6,>=latex,rounded corners,baseline=-1cm,rotate=-90]
\draw[->] (2.7,0) node[anchor=north] {$1$} -- (0,0) -- (0,0.7) node[anchor=west] {$0$};
\draw[->] (2.7,-1) node[anchor=north] {$0$} -- (1,-1) -- (1,0.7) node[anchor=west] {$0$};
\draw[->] (2.7,-2) node[anchor=north] {$1$} -- (2,-2) -- (2,0.7) node[anchor=west] {$1$};
\draw (1.5,0) node[anchor=west,scale=.7] {$0$};
\draw (1.5,-1) node[anchor=west,scale=.7] {$0$};
\draw (2,-.5) node[anchor=north,scale=.7] {$0$};
\end{tikzpicture}
z^2
&
\begin{tikzpicture}[scale=.6,>=latex,rounded corners,baseline=-1cm,rotate=-90]
\draw[->] (2.7,0) node[anchor=north] {$1$} -- (0,0) -- (0,0.7) node[anchor=west] {$0$};
\draw[->] (2.7,-1) node[anchor=north] {$1$} -- (1,-1) -- (1,0.7) node[anchor=west] {$0$};
\draw[->] (2.7,-2) node[anchor=north] {$0$} -- (2,-2) -- (2,0.7) node[anchor=west] {$1$};
\draw (1.5,0) node[anchor=west,scale=.7] {$0$};
\draw (1.5,-1) node[anchor=west,scale=.7] {$0$};
\draw (2,-.5) node[anchor=north,scale=.7] {$1$};
\end{tikzpicture}
z^2&
\begin{tikzpicture}[scale=.6,>=latex,rounded corners,baseline=-1cm,rotate=-90]
\draw[->] (2.7,0) node[anchor=north] {$1$} -- (0,0) -- (0,0.7) node[anchor=west] {$0$};
\draw[->] (2.7,-1) node[anchor=north] {$1$} -- (1,-1) -- (1,0.7) node[anchor=west] {$0$};
\draw[->] (2.7,-2) node[anchor=north] {$0$} -- (2,-2) -- (2,0.7) node[anchor=west] {$1$};
\draw (1.5,0) node[anchor=west,scale=.7] {$0$};
\draw (1.5,-1) node[anchor=west,scale=.7] {$1$};
\draw (2,-.5) node[anchor=north,scale=.7] {$0$};
\end{tikzpicture}
z^2
&
\begin{tikzpicture}[scale=.6,>=latex,rounded corners,baseline=-1cm,rotate=-90]
\draw[->] (2.7,0) node[anchor=north] {$1$} -- (0,0) -- (0,0.7) node[anchor=west] {$0$};
\draw[->] (2.7,-1) node[anchor=north] {$1$} -- (1,-1) -- (1,0.7) node[anchor=west] {$0$};
\draw[->] (2.7,-2) node[anchor=north] {$1$} -- (2,-2) -- (2,0.7) node[anchor=west] {$1$};
\draw (1.5,0) node[anchor=west,scale=.7] {$0$};
\draw (1.5,-1) node[anchor=west,scale=.7] {$0$};
\draw (2,-.5) node[anchor=north,scale=.7] {$1$};
\end{tikzpicture}
z^3
\\
z \aaa_2^- & z^2 \aaa_2^- \aaa_3^+ & z^2 \aaa_3^- & z^2 \aaa_1^+ \aaa_2^- \kk_3 & z^3.
\end{array}
\end{align*}
\end{example}

\subsection{Recursion relation of $X_\alpha(z)$}

For $0\le i\le n-1,0\le j\le n$, define $T(z)_{ij}$ to be the $t$-oscillator valued weight for the following configuration:
\begin{equation}
T(z)_{ij} = 
\begin{tikzpicture}[scale=.6,>=latex, baseline=3.5cm]
\draw[-] (0,-1) -- (0,1);
\draw[-] (0,2) -- (0,5);
\draw[-] (0,6) -- (0,9);
\draw[->,rounded corners] (0,10) -- (0,12) -- (1,12) node[anchor=west] {$0$};
\foreach \y in {1,5,9}
	\draw[dashed] (0,\y) -- (0,\y+1);
\foreach \y in {0,3,7,11}
	\draw[->] (-1,\y) node[anchor=east] {$0$} -- (1,\y) node[anchor=west] {$0$};
\draw[->] (-1,4) node[anchor=east] {$0$} -- (1,4) node[anchor=west] {$1 \quad \leftarrow j$};
\draw[->] (-1,8) node[anchor=east] {$i \rightarrow \quad 1$} -- (1,8) node[anchor=west] {$0$};
\fill[black] (0,-.5) circle (.1) node[anchor=east, scale=.8] {$z^{\hh}$};
\foreach \y in {0,1,3,4}
	\draw (0,\y-.5) node[anchor=west, scale=.7] {$1$};
\foreach \y in {5,7,8,9,11}
	\draw (0,\y-.5) node[anchor=west, scale=.7] {$0$};
\end{tikzpicture}
\end{equation}
It is a column consisting of $(n-1)$ vertices, and the $t$-oscillators attached to the $r$th one from the top is 
denoted by $\ap_r, \am_r, \kk_r$. 
The indices $i$ (resp.\ $j$) specifies the position of the unique 1 on the left (resp.\ right) 
horizontal edges, and $i=0$ (resp.\ $j=0$) means  
that all the left (resp.\ right) horizontal edges assume 0. 
The figure corresponds to a case $i,j \ge 1$.
Explicitly one has 
\begin{equation}\label{a:Tdef}
T(z)_{i0}=
\ap_i \; (0 \le i \le n-1) \quad \text{ and }\quad 
T(z)_{ij} = \begin{cases}
z\kk_j\cdots \kk_{n-1}&(j=i+1)\\
z\ap_i\am_{j-1}\kk_j\cdots \kk_{n-1}&(j\ge i+2)\\
0&(j\le i)
\end{cases}
\quad \text{ for }j\ge1,
\end{equation}
where we regard $\aaa^+_0$ as 1.
Note that $T(z)_{ij}$ depends on $\apm_r, \kk_r$ with $r=1,\ldots, n-1$.

\begin{proposition}[Recursion relation of $X_\alpha(z)$ with respect to rank]\label{a:pr:rt}
Let $X_0(z), \ldots, X_n(z)$ be the operators defined in~\eqref{a:ef12} for the $n$-ASEP.
Let $\widetilde{X}_0(z), \ldots, \widetilde{X}_{n-1}(z)$ be those for the $(n\!-\!1)$-ASEP with 
$\aaa^\pm_i, \kk_i$ relabelled as $\aaa^\pm_{i+n-1}, \kk_{i+n-1}$.
Then the following recursion relation is valid:
\begin{align}\label{a:xxt}
X_\alpha(z) = \sum_{i=0}^{n-1}\widetilde{X}_i(z) T(z)_{i\alpha}\qquad (0 \le \alpha \le n).
\end{align}
\end{proposition}

\begin{proof}
Consider the diagram~\eqref{a:ef12}, and let $s_\beta\, (\beta=1,\ldots, n-1)$ be the variable on the left side of the crossing for the $\beta$th horizontal edge from the top in the rightmost column.
As demonstrated in Example~\ref{a:ex:15},  a close inspection of the rightmost two columns in~\eqref{a:ef12} shows that the configurations that make non-zero contributions to $X_\alpha(z)$ are only those satisfying $s_\beta = \delta_{\beta, i}$ for some $i=0,1,\ldots, \alpha-1$. 
The corresponding term is equal to $\widetilde{X}_i(z) T(z)_{i\alpha}$, where the factors $T(z)_{i\alpha}$ is the weight from the rightmost  $(n\!-\!1)$ vertices and $\widetilde{X}_i(z)$ emerges from the remaining part. 
 \end{proof}
 
 In view of (\ref{a:Tdef}), the sum (\ref{a:xxt}) for the case $\alpha \ge 1$ is actually restricted as
 $\sum_{i=0}^{\alpha-1}\widetilde{X}_i(z) T(z)_{i\alpha}$.
Clearly, $\widetilde{X}_i(z)$ and $T(z)_{i\alpha}$ in (\ref{a:xxt}) are commutative. 

\begin{example}\label{a:ex:15}
Let us consider $X_{\alpha=7}(z)$ for $n=9$ given as a configuration sum~\eqref{a:ef12}. 
We illustrate how the $i=3$ term in the RHS of~\eqref{a:xxt} shows up.
A crucial property is the strange weight conservation of  $S^{ab}_{ij}$ in~\eqref{t:s5V}. 
\newcommand{\basepicture}[1]{
\begin{tikzpicture}[scale=.7,>=latex]
\draw[->,rounded corners] (0,-1) -- (0,7) -- (2,7) node[anchor=west] {$0$};
\draw[->,rounded corners] (1,-1) -- (1,8) -- (2,8) node[anchor=west] {$0$};
\foreach \y in {0,1,3,4,5,6}
	\draw[->] (-1,\y) -- (2,\y) node[anchor=west] {$0$};
\draw[->] (-1,2) -- (2,2) node[anchor=west] {$1$};
\fill[black] (0,-.5) circle (.1);
\fill[black] (1,-.5) circle (.1);
#1
\end{tikzpicture}
}
\[
\basepicture{}
\qquad
\basepicture{
\foreach \y in {3,4,...,7}
	\draw (1,\y-.5) node[anchor=west,scale=.7] {$0$};
\foreach \y in {0,1,2}
	\draw (1,\y-.5) node[color=red,anchor=west,scale=.7] {$1$};
}
\qquad
\basepicture{
\foreach \y in {3,4,...,7}
	\draw (1,\y-.5) node[anchor=west,scale=.7] {$0$};
\foreach \y in {0,1,2}
	\draw (1,\y-.5) node[anchor=west,scale=.7] {$1$};
\foreach \y in {0,1,2}
	\draw (.5,\y) node[color=red,anchor=north,scale=.7] {$0$};
}
\qquad
\basepicture{
\foreach \y in {3,4,...,7}
	\draw (1,\y-.5) node[anchor=west,scale=.7] {$0$};
\foreach \y in {0,1,2}
	\draw (1,\y-.5) node[anchor=west,scale=.7] {$1$};
\foreach \y in {0,1,2}
	\draw (.5,\y) node[anchor=north,scale=.7] {$0$};
\foreach \y in {0,1,...,5}
	\draw (0,\y-.5) node[color=red,anchor=west,scale=.7] {$1$};
\draw (0,5.5) node[color=red,anchor=west,scale=.7] {$0$};
\draw (.5,5) node[color=blue,anchor=north,scale=.7] {$\underline{1}$};
\draw (.5,6) node[color=red,anchor=north,scale=.7] {$0$};
\draw (.5,7) node[color=red,anchor=north,scale=.7] {$0$};
}
\qquad
\basepicture{
\foreach \y in {3,4,...,7}
	\draw (1,\y-.5) node[anchor=west,scale=.7] {$0$};
\foreach \y in {0,1,2}
	\draw (1,\y-.5) node[anchor=west,scale=.7] {$1$};
\foreach \y in {0,1,2,6,7}
	\draw (.5,\y) node[anchor=north,scale=.7] {$0$};
\foreach \y in {0,1,...,5}
	\draw (0,\y-.5) node[anchor=west,scale=.7] {$1$};
\draw (0,5.5) node[anchor=west,scale=.7] {$0$};
\draw (.5,5) node[color=blue,anchor=north,scale=.7] {$\underline{1}$};
\draw (.5,4) node[color=red,anchor=north,scale=.7] {$0$};
\draw (.5,3) node[color=red,anchor=north,scale=.7] {$0$};
}
\]
From left to right, we are performing the following steps:
\begin{enumerate}
\item The rightmost two columns in the diagram for $X_7(z)$.
\item Red vertical edge variables are determined.
\item Red horizontal edge variables are determined.
\item{\em Assuming} that the highest non-zero variable is the blue underlined one, the variables shown in red are determined.
\item  Red horizontal edge variables are determined.
\end{enumerate}
The weight of the eight rightmost vertices and $\bullet$ is $z \aaa^+_3\aaa^-_6 \kk_7 \kk_8 = T(z)_{37}$.
The remaining part of the configuration sum can be identified with  $\widetilde{X}_3(z)$.
\end{example}

\begin{example}
A diagrammatic representation of (\ref{a:xxt}) for $n=3$ is as follows:
\begin{align*}
\raisebox{9pt}{$X_0(z)$} & \begin{array}{@{\;=}c@{\;+\;}c@{\;+\;}c} \begin{tikzpicture}[scale=.6,>=latex,rounded corners,baseline=-1cm,rotate=-90]
\draw[->] (3,0) node[anchor=north] {$0$} -- (0,0) -- (0,1) node[anchor=west] {$0$};
\draw[->] (3,-1) node[anchor=north] {} -- (1,-1) -- (1,1) node[anchor=west] {$0$};
\draw[->] (3,-2) node[anchor=north] {} -- (2,-2) -- (2,1) node[anchor=west] {$0$};
\draw (1.5,0) node[anchor=west,scale=.7] {$0$};
\draw (1.5,-1) node[anchor=west,scale=.7] {$0$};
\draw (2,-.5) node[anchor=north,scale=.7] {$0$};
\fill[black] (2.6,-2) circle (.1);
\fill[black] (2.6,-1) circle (.1);
\end{tikzpicture}
&
\begin{tikzpicture}[scale=.6,>=latex,rounded corners,baseline=-1cm,rotate=-90]
\draw[->] (3,0) node[anchor=north] {$0$} -- (0,0) -- (0,1) node[anchor=west] {$0$};
\draw[->] (3,-1) node[anchor=north] {} -- (1,-1) -- (1,1) node[anchor=west] {$0$};
\draw[->] (3,-2) node[anchor=north] {} -- (2,-2) -- (2,1) node[anchor=west] {$0$};
\draw (1.5,0) node[anchor=west,scale=.7] {$0$};
\draw (1.5,-1) node[anchor=west,scale=.7] {$1$};
\draw (2,-.5) node[anchor=north,scale=.7] {$0$};
\fill[black] (2.6,-2) circle (.1);
\fill[black] (2.6,-1) circle (.1);
\end{tikzpicture}
&
\begin{tikzpicture}[scale=.6,>=latex,rounded corners,baseline=-1cm,rotate=-90]
\draw[->] (3,0) node[anchor=north] {$0$} -- (0,0) -- (0,1) node[anchor=west] {$0$};
\draw[->] (3,-1) node[anchor=north] {} -- (1,-1) -- (1,1) node[anchor=west] {$0$};
\draw[->] (3,-2) node[anchor=north] {} -- (2,-2) -- (2,1) node[anchor=west] {$0$};
\draw (1.5,0) node[anchor=west,scale=.7] {$0$};
\draw (1.5,-1) node[anchor=west,scale=.7] {$0$};
\draw (2,-.5) node[anchor=north,scale=.7] {$1$};
\fill[black] (2.6,-2) circle (.1);
\fill[black] (2.6,-1) circle (.1);
\end{tikzpicture}
\\
\widetilde{X}_0(z) & \aaa_1^+ \widetilde{X}_1(z) & \aaa_2^+ \widetilde{X}_2(z),
\end{array}
\allowdisplaybreaks \\
\raisebox{9pt}{$X_1(z)$} & \begin{array}{@{\;=}c}  \begin{tikzpicture}[scale=.6,>=latex,rounded corners,baseline=-1cm,rotate=-90]
\draw[->] (3,0) node[anchor=north] {$1$} -- (0,0) -- (0,1) node[anchor=west] {$1$};
\draw[->] (3,-1) node[anchor=north] {} -- (1,-1) -- (1,1) node[anchor=west] {$0$};
\draw[->] (3,-2) node[anchor=north] {} -- (2,-2) -- (2,1) node[anchor=west] {$0$};
\draw (1.5,0) node[anchor=west,scale=.7] {$1$};
\draw (1.5,-1) node[anchor=west,scale=.7] {$0$};
\draw (2,-.5) node[anchor=north,scale=.7] {$0$};
\fill[black] (2.6,-2) circle (.1);
\fill[black] (2.6,-1) circle (.1);
\end{tikzpicture}
z
\\
z \kk_1 \kk_2 \widetilde{X}_0(z),
\end{array}
\allowdisplaybreaks \\
\raisebox{9pt}{$X_2(z)$} & \begin{array}{@{\;=}c@{\;+\;}c} \begin{tikzpicture}[scale=.6,>=latex,rounded corners,baseline=-1cm,rotate=-90]
\draw[->] (3,0) node[anchor=north] {$1$} -- (0,0) -- (0,1) node[anchor=west] {$0$};
\draw[->] (3,-1) node[anchor=north] {} -- (1,-1) -- (1,1) node[anchor=west] {$1$};
\draw[->] (3,-2) node[anchor=north] {} -- (2,-2) -- (2,1) node[anchor=west] {$0$};
\draw (1.5,0) node[anchor=west,scale=.7] {$1$};
\draw (1.5,-1) node[anchor=west,scale=.7] {$0$};
\draw (2,-.5) node[anchor=north,scale=.7] {$0$};
\fill[black] (2.6,-2) circle (.1);
\fill[black] (2.6,-1) circle (.1);
\end{tikzpicture}
z
&
\begin{tikzpicture}[scale=.6,>=latex,rounded corners,baseline=-1cm,rotate=-90]
\draw[->] (3,0) node[anchor=north] {$1$} -- (0,0) -- (0,1) node[anchor=west] {$0$};
\draw[->] (3,-1) node[anchor=north] {} -- (1,-1) -- (1,1) node[anchor=west] {$1$};
\draw[->] (3,-2) node[anchor=north] {} -- (2,-2) -- (2,1) node[anchor=west] {$0$};
\draw (1.5,0) node[anchor=west,scale=.7] {$1$};
\draw (1.5,-1) node[anchor=west,scale=.7] {$1$};
\draw (2,-.5) node[anchor=north,scale=.7] {$0$};
\fill[black] (2.6,-2) circle (.1);
\fill[black] (2.6,-1) circle (.1);
\end{tikzpicture}
z
\\
z \aaa_1^- \kk_2 \widetilde{X}_0(z) & z \kk_2 \widetilde{X}_1(z),
\end{array}
\allowdisplaybreaks \\
\raisebox{9pt}{$X_3(z)$} & \begin{array}{@{\;=}c@{\;+\;}c@{\;+\;}c} \begin{tikzpicture}[scale=.6,>=latex,rounded corners,baseline=-1cm,rotate=-90]
\draw[->] (3,0) node[anchor=north] {$1$} -- (0,0) -- (0,1) node[anchor=west] {$0$};
\draw[->] (3,-1) node[anchor=north] {} -- (1,-1) -- (1,1) node[anchor=west] {$0$};
\draw[->] (3,-2) node[anchor=north] {} -- (2,-2) -- (2,1) node[anchor=west] {$1$};
\draw (1.5,0) node[anchor=west,scale=.7] {$0$};
\draw (1.5,-1) node[anchor=west,scale=.7] {$0$};
\draw (2,-.5) node[anchor=north,scale=.7] {$0$};
\fill[black] (2.6,-2) circle (.1);
\fill[black] (2.6,-1) circle (.1);
\end{tikzpicture}
z
&
\begin{tikzpicture}[scale=.6,>=latex,rounded corners,baseline=-1cm,rotate=-90]
\draw[->] (3,0) node[anchor=north] {$1$} -- (0,0) -- (0,1) node[anchor=west] {$0$};
\draw[->] (3,-1) node[anchor=north] {} -- (1,-1) -- (1,1) node[anchor=west] {$0$};
\draw[->] (3,-2) node[anchor=north] {} -- (2,-2) -- (2,1) node[anchor=west] {$1$};
\draw (1.5,0) node[anchor=west,scale=.7] {$0$};
\draw (1.5,-1) node[anchor=west,scale=.7] {$1$};
\draw (2,-.5) node[anchor=north,scale=.7] {$0$};
\fill[black] (2.6,-2) circle (.1);
\fill[black] (2.6,-1) circle (.1);
\end{tikzpicture}
z
&
\begin{tikzpicture}[scale=.6,>=latex,rounded corners,baseline=-1cm,rotate=-90]
\draw[->] (3,0) node[anchor=north] {$1$} -- (0,0) -- (0,1) node[anchor=west] {$0$};
\draw[->] (3,-1) node[anchor=north] {} -- (1,-1) -- (1,1) node[anchor=west] {$0$};
\draw[->] (3,-2) node[anchor=north] {} -- (2,-2) -- (2,1) node[anchor=west] {$1$};
\draw (1.5,0) node[anchor=west,scale=.7] {$0$};
\draw (1.5,-1) node[anchor=west,scale=.7] {$0$};
\draw (2,-.5) node[anchor=north,scale=.7] {$1$};
\fill[black] (2.6,-2) circle (.1);
\fill[black] (2.6,-1) circle (.1);
\end{tikzpicture}
z
\\
z \aaa_2^- \widetilde{X}_0(z) & z \aaa_1^+ \aaa_2^- \widetilde{X}_1(z) & z \widetilde{X}_2(z).
\end{array}
\end{align*}
One can also check them directly by using Example~\ref{a:ex:X}.
A matrix form of (\ref{a:xxt}) is 
\begin{align}
(X_0(z), X_1(z), X_2(z), X_3(z)) = 
(\widetilde{X}_0(z), \widetilde{X}_1(z), \widetilde{X}_2(z))
\begin{pmatrix}\label{a:rec3}
1 & z \kk_1\kk_2 & z \am_1\kk_2 & z \am_2
\\
\ap_1 & 0 & z \kk_2 & z \ap_1 \am_2
\\
\ap_2 & 0 & 0 & z
\end{pmatrix}.
\end{align}
The transposition of the matrix here reproduces the $4 \times 3$ matrix in~\cite[Eq.(73)]{CDW}
under the conventional change $(\ap_i, \am_i, \kk_i, z) \rightarrow (a_{i+1}, a^\dagger_{i+1},k_{i+1},z)$.
The matrix $(T(z)_{i j})_{0 \le i \le n-1, 0 \le j \le n}$ for $n=4$ is available in~\eqref{a:ex:t4}.
The recursion relation of the matrix product operators of the form~\eqref{a:xxt} without a spectral parameter appeared earlier in~\cite{PEM}.
\end{example}

\begin{remark}\label{a:re:nn}
The diagram (\ref{a:ef12}) makes sense  only for $n \ge 2$.
We extend the definition of $X_0(z),\ldots, X_n(z)$ to $n=0,1$ as follows:
\begin{equation}
\begin{split}
\text{$n=1$:} \;\;&X_0(z) = 1,\quad X_1(z) = z,
\\
\text{$n=0$:} \; \; & X_0(z) = 1. 
\end{split}
\end{equation}
Then the recursion relation (\ref{a:xxt}) also holds at $n=1$. 
The $n=1$ case agrees with~\cite[Eq.(42)]{CDW}.
\end{remark}

Several diagrammatic representations of the matrix product operators $X_\alpha$ 
or $X_\alpha(z)$ were devised in earlier works~\cite{CDW,PEM}. 
However, we find that the CTM representation in (\ref{a:ef12}) is the simplest 
and most systematic, offering a clear visualization of their evaluation and 
providing clarity to Proposition~\ref{a:pr:rt}. 
This advancement has been made possible through the introduction of the 
strange five vertex model.

\section{Proof of the Zamolodchikov--Faddeev algebra relation}\label{sec:zf}

The aim of this section is to show the ZF algebra relation~\eqref{a:zf}. 

\subsection{$RLL=LLR$ relation}
Recall that the bosonic Fock space $F$ is defined in (\ref{a:F}).
Set $\F=F^{\ot (n+1)}$. 
For a sequence of non-negative integers $\mm=(m_0,\ldots,m_n)$,\footnote{This ${\bf m}$ should not be confused with 
the multiplicity introduced in (\ref{a:VP}).}
we write $\ket{\mm}=\ket{m_0}\ot\cdots\ot\ket{m_n} \in \F$.
Set 
\begin{align}
\F_l & =\Q(t)\langle\ket{\mm}\mid |  \mm \in D_l \rangle,
\\
D_l &= \{\mm = (m_0,\ldots, m_n) \in (\Z_{\ge 0})^{n+1}\mid m_0+\cdots +  m_n = l\}.
\label{a:dl}
\end{align}
Denote by ${\bf e}_i \in D_{l=1}$ the $i$th standard basis vector.
For $0\le\alpha,\beta\le n$
and $\aaa, \bb \in D_l$,  set
\begin{align}
L(z)^{\beta, \bb}_{\alpha, \aaa} = \delta_{ {\bf e}_\alpha +\aaa}^{ {\bf e}_\beta+ \bb}\,
t^{a_{\beta+1}+ \cdots + a_n}(1-t^{a_\beta}z^{\theta(\alpha=\beta)})z^{\theta(\alpha>\beta)},
\label{a:L1}
\end{align}
and define a linear operator 
\begin{equation}\label{a:L}
\begin{aligned}
L(z)\colon \quad  \F_1 \otimes \F_l & \longrightarrow \F_1 \otimes \F_l
\\
\ket{\eb_\alpha} \otimes \ket{\aaa} & \longmapsto
\sum_{\beta\in \{0,\ldots, n\}, \bb \in D_l}L(z)^{\beta, \bb}_{\alpha, \aaa}
\ket{\eb_\beta} \otimes \ket{\bb}.
\end{aligned}
\end{equation}
The dependence on $t$ and $l$ has been suppressed in the notation.
We also introduce the components
\begin{equation}\label{a:Lab}
\begin{aligned}
L(z)_\alpha^\beta \colon \;\; \F_l & \longrightarrow  \F_l\quad (0\le \alpha, \beta \le n)
\\
\ket{\aaa} & \longmapsto
\sum_{\bb \in D_l}L(z)^{\beta, \bb}_{\alpha, \aaa}\ket{\bb} = 
L(z)^{\beta, \aaa+\eb_\alpha-\eb_\beta}_{\alpha, \aaa}\ket{\aaa+\eb_\alpha-\eb_\beta},
\end{aligned}
\end{equation}
where the RHS is to be understood as 0 unless $\aaa+{\bf e}_\alpha-{\bf e}_{\beta} \in D_l$.

\begin{remark}\label{a:re:L}
Let $R_{\rm KMMO}(z)^{e_k, \delta}_{e_j, \beta}$ denote the elements 
of the $R$ matrix given in the first equation in~\cite[App.A]{KMMO} with 
$m$ replaced with $l$.\footnote{Indices $1,\ldots, n+1$ in~\cite{KMMO} should also be 
replaced with $0,\ldots, n$.}
According to~\cite[Eqs.(15),(16)]{KMMO},  its stochastic gauge is given by 
$S_{\rm KMMO}(z)^{e_k, \delta}_{e_j, \beta}:=
q^\eta R_{\rm KMMO}(z)^{e_k, \delta}_{e_j, \beta}$ with 
$\eta = \delta_0+\cdots + \delta_{k-1} - (\beta_{j+1}+ \cdots + \beta_n)$.
Then one has  
\begin{align}\label{a:SK}
(1-q^{2l}z)S_{\rm KMMO}(q^{1-l}z^{-1})^{\beta, \bb}_{\alpha, \aaa}
=
\delta_{ {\bf e}_\alpha +\aaa}^{ {\bf e}_\beta+ \bb}\,q^{2(a_0+ \cdots + a_{\beta-1})}(1-q^{2a_\beta}z^{\theta(\alpha=\beta)})z^{\theta(\alpha<\beta)}.
\end{align}
The element (\ref{a:L1}) is obtained from (\ref{a:SK})  by reversing the indices as 
$(\alpha, \beta) \rightarrow (n-\alpha, n-\beta)$, 
$\aaa \rightarrow (a_n,\ldots, a_0)$, 
$\bb \rightarrow (b_n,\ldots, b_0)$ and setting $q^2 \rightarrow t$.
When $l=1$, it reduces to the scalar multiple of (\ref{a:rmat}) 
as $L(z)^{\gamma, \eb_\delta}_{\alpha, \eb_\beta}=(1-tz)R(z)^{\gamma, \delta}_{\alpha, \beta}$.
The operator $L(z)$ is ``stochastic" in the sense that 
$\sum_{\beta, \bb} L(z)^{\beta, \bb}_{\alpha, \aaa} = 1-z t^l$ is independent of 
$\alpha$ and $\aaa$.
\end{remark}

\begin{proposition}[$RLL=LLR$ relation]\label{a:pr:RLL}
For any $a,b,i,j \in \{0,\ldots, n\}$, the following equality is valid:
\begin{equation}\label{a:rll}
\sum_{a',b'=0}^nR(x/y)_{i, j}^{a', b'}L(y)_{b'}^bL(x)_{a'}^a
=\sum_{i',j'=0}^nL(x)_i^{i'}L(y)_j^{j'}R(x/y)_{i', j'}^{a, b}.
\end{equation}
\end{proposition}
\begin{proof}
From Remark~\ref{a:re:L}, one has the $RLL=LLR$ relation
$R_{12}(x/y)L_{13}(x)L_{23}(y) = L_{23}(y)L_{13}(x)R_{12}(x/y)$ 
in $\mathrm{End}(\F_1 \otimes \F_1 \otimes \F_l)$
by  setting $(k,l,m) \rightarrow (1,1,l)$ and 
$(x,y)\rightarrow (x/y,y)$ in the Yang--Baxter equation for the stochastic $R$'s 
in~\cite[Prop.4]{KMMO}.
The relation (\ref{a:rll}) is a component of it corresponding to the transition
$\ket{{\bf e}_i} \otimes \ket{{\bf e}_j}  \mapsto \ket{{\bf e}_a} \otimes \ket{{\bf e}_b}$ 
in the $\F_1 \otimes \F_1$ part.
\end{proof}

Consider the constant part of (\ref{a:Lab}):
\begin{equation} \label{m:Lt}
L(0)_\alpha^\beta \ket{{\bf m}} = \begin{cases}
t^{m_{\beta+1}+\cdots+m_n}\ket{{\bf m}}&(\alpha=\beta),\\
t^{m_{\beta+1}+\cdots+m_n}(1-t^{m_\beta}) \ket{{\bf m}+{\bf e}_\alpha-{\bf e}_\beta} &(\alpha<\beta),\\
0&(\alpha>\beta).
\end{cases}
\end{equation}
In Proposition~\ref{a:pr:RLL}, replace $(x,y)$ by $(cx^{-1},cy^{-1})$ and take the limit $c\rightarrow 0$.
The result reads
\begin{equation} \label{m:RLtLt}
\sum_{a',b'=0}^nR(y/x)_{i, j}^{a', b'}L(0)_{b'}^bL(0)_{a'}^a
=\sum_{i',j'=0}^nL(0)_i^{i'}L(0)_j^{j'}R(y/x)_{i', j'}^{a, b}.
\end{equation}

We would like to reformulate \eqref{m:RLtLt} as a $t$-oscillator valued equation. 
On $F$, $t$-oscillators $\kk,\ap,\am$ act as~\eqref{a:F}.
We let $n+1$ copies of the $t$-oscillators act on $\F$. 
We distinguish them by putting a subscript
as $\ap_i$ to signify which component it acts on. Note that we label the components as
$0,1,\ldots,n$ from the left. 
Introduce further the space without the $0$~th component in $\F$ and the projection onto it as 
\begin{align}
\overline{\F} &= \bigoplus_{(m_1,\ldots, m_n) \in (\Z_{\ge 0})^n}
\mathbb{Q}(t) |m_1\rangle \otimes \cdots \otimes |m_n\rangle,
\\
\iota \colon & \;  \F \rightarrow \overline{\F}\,;\;\;  
|m_0\rangle \otimes |m_1\rangle \otimes \cdots \otimes |m_n\rangle  
\mapsto  |m_1\rangle \otimes \cdots \otimes |m_n\rangle.
\end{align}
Then by taking the image of (\ref{m:RLtLt}) by $\iota$, we obtain
\begin{equation} \label{a:nr}
\sum_{a',b'=0}^nR(y/x)_{i, j}^{a', b'}\mathcal{L}_{b'}^b\mathcal{L}_{a'}^a
=\sum_{i',j'=0}^n\mathcal{L}_i^{i'}\mathcal{L}_j^{j'}R(y/x)_{i', j'}^{a, b},
\end{equation}
where $\mathcal{L}_\alpha^\beta \in \mathrm{End}(\overline{\F})$ is given 
in terms of the $t$-oscillators $\apm_i,\kk_i \, (i=1,\ldots, n)$ as
\begin{equation}\label{a:cL}
\mathcal{L}_\alpha^\beta= \begin{cases}
\kk_{\beta+1}\cdots\kk_n&(\alpha=\beta),\\
\ap_\alpha\am_\beta\kk_{\beta+1}\cdots\kk_n&(\alpha<\beta),\\
0&(\alpha>\beta),
\end{cases}
\end{equation}
for  $0 \le \alpha, \beta \le n$ with $\ap_0=1$ in the middle case.
As mentioned in Remark~\ref{a:re:tau}, 
the formula (\ref{a:cL}) can be interpreted as  the ``Holstein--Primakov representation'' of the $L$-operators. 
Analogous results in the ``crystal'' gauge as opposed to the stochastic one adopted here 
have been obtained in~\cite{IKO} for $A^{(1)}_n$ and also $D^{(1)}_n$
in the study of quantized box-ball systems.

\subsection{Rank reducing $RTT=TTR$ relation}
Recall that \( T(z)_{\alpha \beta} \) is defined in (\ref{a:Tdef})
 for \( 0 \leq \alpha \leq n-1 \) and \( 0 \leq \beta \leq n \).
By direct comparison with (\ref{a:cL}), we obtain the relation
\begin{equation} \label{m:LtT}
\mathcal{L}_\alpha^\beta = 
T(z)_{\alpha,\beta+1} 
(\am_n)^{\delta_{\beta n}} (z^{-1}\kk_n)^{\theta(\beta \neq n)}
\quad (0 \leq \alpha \leq n-1, \, 0 \leq \beta \leq n),
\end{equation}
where \( T(z)_{\alpha, n+1} := T(z)_{\alpha 0} \).
Operator ordering does not matter in (\ref{m:LtT}) 
since \( T(z)_{ij} \) involves only \( \apm_\alpha \) and \( \kk_\alpha \) 
with \( \alpha = 1, \ldots, n-1 \).
The relation (\ref{m:LtT}) serves as the crucial link between the $RLL=LLR$ relation (Proposition~\ref{a:pr:RLL}) 
and the forthcoming rank-reducing $RTT=TTR$ relation (Proposition~\ref{a:pr:rr}), 
a connection that has previously gone unnoticed in the literature.
\begin{example}\label{a:ex:L}
From (\ref{a:cL}), the LHS of (\ref{m:LtT})  for $n=4$ in the matrix form is given as
\begin{equation}
\bigl(\mathcal{L}_\alpha^\beta\bigr)_{0 \le \alpha \le 3, 0 \le \beta \le 4}
=\begin{pmatrix}
 \kk_1 \kk_2 \kk_3 \kk_4 &  \am_1\kk_2\kk_3\kk_4 &  \am_2\kk_3\kk_4 &  \am_3\kk_4 &  \am_4
\\
0 &  \kk_2 \kk_3\kk_4 & \ap_1\am_2 \kk_3\kk_4 &  \ap_1 \am_3 \kk_4 & \ap_1 \am_4
\\
0 & 0 &  \kk_3 \kk_4 &  \ap_2\am_3 \kk_4 &  \ap_2\am_4 
\\
0 & 0 & 0 &  \kk_4 &  \ap_3\am_4
\end{pmatrix}.
\end{equation}
Similarly, the RHS of (\ref{m:LtT}) for $n=4$ reads
\begin{align}
\begin{pmatrix}
T(z)_{01}z^{-1}\kk_4 & T(z)_{02}z^{-1}\kk_4 & T(z)_{03}z^{-1}\kk_4 & T(z)_{04}z^{-1}\kk_4 & T(z)_{00}\am_4
\\
T(z)_{11}z^{-1}\kk_4 & T(z)_{12}z^{-1}\kk_4 & T(z)_{13}z^{-1}\kk_4 & T(z)_{14}z^{-1}\kk_4 & T(z)_{10}\am_4
\\
T(z)_{21}z^{-1}\kk_4 & T(z)_{22}z^{-1}\kk_4 & T(z)_{23}z^{-1}\kk_4 & T(z)_{24}z^{-1}\kk_4 & T(z)_{20}\am_4
\\
T(z)_{31}z^{-1}\kk_4 & T(z)_{32}z^{-1}\kk_4 & T(z)_{33}z^{-1}\kk_4 & T(z)_{34}z^{-1}\kk_4 & T(z)_{30}\am_4
\end{pmatrix}.
\end{align}
They indeed coincide due to (\ref{a:Tdef}) with $n=4$, which reads
\begin{align}\label{a:ex:t4}
(T(z)_{\alpha \beta})_{0 \le \alpha \le 3, 0 \le \beta \le 4}
=\begin{pmatrix}
1 & z \kk_1\kk_2 \kk_3 & z \am_1\kk_2 \kk_3 & z \am_2 \kk_3  & z\am_3
\\
\ap_1 & 0 & z \kk_2\kk_3 & z \ap_1 \am_2 \kk_3 & z \ap_1\am_3
\\
\ap_2 & 0 & 0 & z \kk_3 & z \ap_2\am_3
\\
\ap_3 & 0 & 0 & 0 & z
\end{pmatrix}.
\end{align} 
\end{example}

The relation (\ref{a:nr}) depends on \( x \) and \( y \) 
only through the ratio \( y/x \). 
However, switching to the description with \( T(z)_{\alpha \beta} \) 
via (\ref{m:LtT}) appropriately reinstates a non-trivial dependence 
on the spectral parameters, as we will demonstrate below.

\begin{proposition}[Rank-reducing $RTT=TTR$ relation]\label{a:pr:rr}
For $0 \le a,b \le n$ and $0 \le i,j \le n-1$, we have
\begin{equation}\label{a:rrr}
\sum_{a',b'=0}^{n-1}R(y/x)_{i, j}^{a', b'}T(y)_{b'b}T(x)_{a'a}
=\sum_{i',j'=0}^nT(x)_{ii'}T(y)_{jj'}R(y/x)_{i', j'}^{a, b}.
\end{equation}
\end{proposition}

\begin{proof}
In \eqref{a:nr},  
substitute \eqref{m:LtT} with $z$ taken as $y,x,x,y$ for 
$\mathcal{L}_{b'}^b,  \mathcal{L}_{a'}^a, \mathcal{L}_i^{i'}, \mathcal{L}_j^{j'}$, respectively.
Restrict the range of $i,j$ to $0\le i,j\le n-1$. 
Then from the weight preservation, one can also restrict the summation for $a',b'$ to $0\le a',b'\le n-1$.
Hence we get
\begin{align*}
&\sum_{a',b'=0}^{n-1}R(y/x)_{i, j}^{a', b'}y^{\delta_{bn}}T(y)_{b',b+1}(\am_n)^{\delta_{bn}}\kk_n^{\theta(b\ne n)}
\times x^{\delta_{an}}T(x)_{a',a+1}(\am_n)^{\delta_{an}}\kk_n^{\theta(a\ne n)}\\
=&\sum_{i',j'=0}^nx^{\delta_{i'n}}T(x)_{i,i'+1}(\am_n)^{\delta_{i'n}}\kk_n^{\theta(i'\ne n)}
\times y^{\delta_{j'n}}T(y)_{j,j'+1}(\am_n)^{\delta_{j'n}}\kk_n^{\theta(j'\ne n)}R(y/x)_{i', j'}^{a, b}.
\end{align*}
Sending the $t$-oscillators with subscript $n$ to the right using the commutativity  
with $T(z)_{rs}$, we have
\[
\sum_{a',b'=0}^{n-1}x^{\delta_{an}}y^{\delta_{bn}}R(y/x)_{i, j}^{a', b'}T(y)_{b',b+1}T(x)_{a',a+1} A
=\sum_{i',j'=0}^nx^{\delta_{i'n}}y^{\delta_{j'n}}T(x)_{i,i'+1}T(y)_{j,j'+1} B R(y/x)_{i', j'}^{a, b},
\]
where 
\[
A=(\am_n)^{\delta_{bn}}\kk_n^{\theta(b\ne n)}(\am_n)^{\delta_{an}}\kk_n^{\theta(a\ne n)},\qquad
B=(\am_n)^{\delta_{i'n}}\kk_n^{\theta(i'\ne n)}(\am_n)^{\delta_{j'n}}\kk_n^{\theta(j'\ne n)}.
\]
Using the commutation relation $\kk_n\am_n=t^{-1}\am_n\kk_n$, we obtain
\[
A=t^{-\theta(a=n,b\ne n)}(\am_n)^{\delta_{an}+\delta_{bn}}\kk_n^{\theta(a\ne n)+\theta(b\ne n)},\qquad
B=t^{-\theta(i'\ne n,j'=n)}(\am_n)^{\delta_{i'n}+\delta_{j'n}}\kk_n^{\theta(i'\ne n)+\theta(j'\ne n)}.
\]
The $t$-oscillator parts of $A$ and $B$ are equal when $R(y/x)^{a, b}_{i', j'}\ne0$.
Taking the coefficients of $A$ and $B$ acting on the $n$th component, we arrive at the following:
\begin{align*}
&\sum_{a',b'=0}^{n-1}R(y/x)_{i, j}^{a', b'}T(y)_{b',b}T(x)_{a',a} \\
& \hspace{30pt} = \sum_{i',j'=0}^nT(x)_{i,i'}T(y)_{j,j'} x^{\delta_{i'0}-\delta_{a0}}y^{\delta_{j'0}-\delta_{b0}}
t^{-\theta(i'\ne0,j'=0)+\theta(a=0,b\ne0)}R(y/x)_{i'-1,j'-1}^{a-1,b-1}.
\end{align*}
Note that we have decreased the indices $a,b,i',j'$ by 1.
Additionally note that when $R(y/x)_{i', j'}^{a, b}\neq 0$, one has $\delta_{i'0}-\delta_{a0}=-(\delta_{j'0}-\delta_{b0})$.
Hence, the coefficients in the RHS can be expressed in terms of $z = y/x$ (see~\eqref{eq:qp} below).
Thus, the proof is attributed to the next lemma.
\end{proof}

\begin{lemma}[Quasi-periodicity of the $R$ matrix]
For $a,b,i',j' \in \Z_{n+1}$, the following relation holds
\begin{equation}
\label{eq:qp}
R(z)_{i', j'}^{a, b} = z^{\delta_{j'0}-\delta_{b0}} t^{-\theta(i'\ne0,j'=0)+\theta(a=0,b\ne0)}R(z)_{i'-1,j'-1}^{a-1,b-1},
\end{equation}
where all indices of $R(z)$ should be taken to be in $\{0,\ldots, n\}$ for the inequalities in~\eqref{a:rmat}.
\end{lemma}

\begin{proof}
This is immediately checked from (\ref{a:rmat}) as   
\[
R(z)^{\alpha+1,0}_{\alpha+1,0}=t^{-1}R(z)^{\alpha, n}_{\alpha, n},
\quad
R(z)^{0,\alpha+1}_{0,\alpha+1}=tR(z)^{n,\alpha}_{n,\alpha},
\quad
R(z)^{\alpha+1,0}_{0,\alpha+1} = z^{-1}R(z)^{\alpha, n}_{n,\alpha},
\quad
R(z)_{\alpha+1,0}^{0,\alpha+1} = zR(z)_{\alpha, n}^{n,\alpha}
\]
for $0 \le \alpha \le n-1$.
\end{proof}

\subsection{Proof of the Zamolodchikov--Faddeev algebra relation}

\begin{theorem}\label{a:th:zf}
The set of operators $\{X_0(z), \ldots, X_n(z)\}$ defined by the CTM diagram~\eqref{a:ef12} for $n \ge 2$ and in Remark~\ref{a:re:nn} for $n=0,1$ satisfy the ZF algebra relation~\eqref{a:zf} with the structure function given by the $R$ matrix in~\eqref{a:rmat}.
\end{theorem}

\begin{proof}
We prove (\ref{a:zf})  by induction on $n$.
Substituting the recursion relation (\ref{a:xxt}) into it, we get 
\begin{equation}
\sum_{i,j=0}^{n-1} \widetilde{X}_i(y)\widetilde{X}_j(x) T(y)_{i \alpha}T(x)_{j \beta}
 =\sum_{\gamma, \delta=0}^n  R(y/x)^{\beta, \alpha}_{\gamma, \delta}
 \sum_{\gamma', \delta'=0}^{n-1}
 \widetilde{X}_{\gamma'}(x)\widetilde{X}_{\delta'}(y) 
T(x)_{\gamma', \gamma}T(y)_{\delta', \delta}.
 \end{equation}
By means of the rank reducing $RTT=TTR$ relation (\ref{a:rrr}), 
the sum $\sum_{\gamma, \delta=0}^n$ in the RHS can be taken. 
The result reads
\begin{equation}
\sum_{i,j=0}^{n-1} \widetilde{X}_i(y)\widetilde{X}_j(x) T(y)_{i \alpha}T(x)_{j \beta}
 =
\sum_{\gamma', \delta'=0}^{n-1}
\widetilde{X}_{\gamma'}(x)\widetilde{X}_{\delta'}(y) 
\sum_{i,j=0}^{n-1} 
R(y/x)^{j,i}_{\gamma', \delta'}
 T(y)_{i \alpha}T(x)_{j \beta}.
\end{equation}
This follows from the ZF algebra relation with one lower rank:
\[
\widetilde{X}_i(y) \widetilde{X}_j(x) = \sum_{\gamma', \delta'=0}^{n-1} \widetilde{X}_{\gamma'}(x)\widetilde{X}_{\delta'}(y) R(y/x)^{j,i}_{\gamma', \delta'}.
\]
Therefore the proof reduces to the $n=0$ case, which is straightforward from Remark~\ref{a:re:nn}.
\end{proof}

Up to convention, the matrix product formula (\ref{a:mho}) was first established 
in~\cite{PEM} by a direct, albeit quite tedious, verification of the hat relation 
(\ref{a:hr1}). The proof was later simplified in~\cite{CDW} by introducing 
the Yang--Baxterizations \(X_0(z), \ldots, X_n(z)\) and the ZF algebra, 
where Proposition~\ref{a:pr:rr} was also derived based on a few lemmas; 
see~\cite[Eq.(36)]{CDW}. Our proof, however, is the most intrinsic from the 
perspective of quantum integrable systems, as it directly stems from the 
Yang--Baxter equation for the stochastic \(R\) matrices, highlighted 
by the key connection (\ref{m:LtT}).

\section{Concluding remarks}\label{sec:cr}

We have unveiled several new insights into the construction of the 
stationary states in the multispecies ASEP by invoking the strange five vertex model.
Let us conclude the paper with two remarks.

(i) The strange five vertex model in this paper differs from~\cite{KMO1,KMO2} 
(see Remark~\ref{a:re:st}), where 
the totally asymmetric simple exclusion process (TASEP) corresponding to $t=0$
was treated using CTMs and the tetrahedron equation.
In particular the tetrahedron equation leads to an elegant proof of the ZF algebra relation 
without going through the inductive steps on the rank $n$~\cite{KMO2,K22}.
It remains an open question whether such a superior variant of the quantum 
oscillator/CTM approach can be formulated that smoothly interpolates 
TASEP and ASEP.
Regarding this issue, an alternative queueing construction, 
as discussed in \cite[Sec.7]{Martin20}, may offer valuable insights.

(ii) Stationary probabilities are connected to Macdonald polynomials,
particularly when an additional parameter $q$ 
and the weight variables are supplemented~\cite{CDW,CMW,DW}.
We have touched upon the implications of the parameter $q$ 
in relation to MLQs and $t$-oscillators in this paper.
However, a comprehensive treatment of their applications 
requires further investigation.

\section*{Acknowledgments}

The authors thank Arvind Ayyer for useful communications.

A.K.\ was supported by Grants-in-Aid for Scientific Research No.~24K06882 from JSPS.
M.O.\ was partly supported by MEXT Promotion of Distinctive Joint Research Center Program JPMXP0723833165.
T.S.\ was supported by Grant-in-Aid for Scientific Research for Early-Career Scientists 23K12983.

\bibliographystyle{alpha}
%\bibliography{asep_refs}{}

\end{document}